\documentclass[10pt, final, journal, letterpaper, twocolumn]{IEEEtran}
\ifCLASSINFOpdf
\else
\fi
\makeatletter
\def\ps@headings{%
\def\@oddhead{\mbox{}\scriptsize\rightmark \hfil \thepage}%
\def\@evenhead{\scriptsize\thepage \hfil \leftmark\mbox{}}%
\def\@oddfoot{}%
\def\@evenfoot{}}
\makeatother 
\usepackage{tikz}
\usepackage{mathtools,amssymb,lipsum}

\usepackage{cuted}
\usepackage{tikz}
\usepackage{tkz-graph}
\usetikzlibrary{arrows,chains,matrix,positioning,scopes,shapes,backgrounds}

\makeatletter
\tikzset{join/.code=\tikzset{after node path={%
\ifx\tikzchainprevious\pgfutil@empty\else(\tikzchainprevious)%
edge[every join]#1(\tikzchaincurrent)\fi}}}

\IEEEoverridecommandlockouts
\usepackage{amsmath}
\usepackage{mathrsfs}
\usepackage{color}
\usepackage{mathrsfs}
\usepackage{amsfonts}
\usepackage{caption}
\usepackage{subfigure}
\usepackage{amsmath}
\usepackage{times}
\usepackage{cite}
\usepackage{lettrine}
\usepackage{amsthm}

\allowdisplaybreaks[4]
\usepackage{array}
\usepackage{amssymb}

\usepackage{stfloats}
\usepackage{graphicx}
\usepackage{footnote}
\usepackage{booktabs}
\usepackage{array}
\usepackage{algorithmic}
\usepackage{algorithm}
\usepackage{subeqnarray}
\usepackage{cases}
\usepackage{threeparttable}
\usepackage{color}
\usepackage{url}
\usepackage{bm}
\usepackage{physics}
\newcommand{\specialcell}[2][c]{%
  \begin{tabular}[#1]{@{}c@{}}#2\end{tabular}}
\newtheorem{theorem}{Theorem}
\newtheorem{lemma}[theorem]{Lemma}

\begin{document}
\bibliographystyle{IEEEtran}

\title{NOMA-Aided Mobile Edge Computing \\via User Cooperation}
\IEEEoverridecommandlockouts

\author{\IEEEauthorblockN{Yuwen Huang, Yuan Liu,~\IEEEmembership{Senior Member,~IEEE}, and Fangjiong Chen}



}
\maketitle
\vspace{-1.5cm}
\maketitle
\begin{abstract}
Exploiting the idle computation resources of mobile devices in mobile edge computing (MEC) system can achieve both channel diversity and computing diversity as mobile devices can offload their computation tasks to nearby mobile devices in addition to MEC server embedded access point (AP). In this paper, we propose a non-orthogonal multiple-access (NOMA)-aided cooperative computing scheme in a basic three-node MEC system consisting of a user, a helper, and an AP. In particular, we assume that the user can simultaneously offload data to the helper and the AP using NOMA, while the helper can locally compute data and offload data to the AP at the same time. We study two optimization problems, energy consumption minimization and offloading data maximization, by joint communication and computation resource allocation of the user and helper. We find the optimal solutions for the two non-convex problems by some proper mathematical methods. Simulation results are presented to demonstrate the effectiveness of the proposed schemes. Some useful insights are provided for practical designs.
\end{abstract}


 \section{Introduction}
 The fifth generation (5G) wireless communication systems are expected to provide ultra-low latency and data hungry services for the mobile devices \cite{Chiang2016}. For example, the augmented reality/virtual reality (AR/VR) applications require a large amount of image and audio information to process within milliseconds latency. Meanwhile, the systems need to accommodate billions of Internet-of-Things (IoT) devices. Consequently, the limitation of mobile devices' battery lives and computation capacities is a crucial challenge for future communication systems. Normally, mobile cloud computing (MCC) has been considered as an efficient and powerful technology to support big data processing by utilizing rich computation resources at the remote cloud centers \cite{Dinh2013}. However, MCC  always suffers the long propagation distances from mobile devices to the cloud centers, resulting in the failure of meeting the critical latency requirements. Fortunately, this problem can be tackled in mobile edge computing (MEC) where the strict latency requirements can be guaranteed  \cite{Chiang2016}. Distinguished from MCC, the MEC servers are dedicatedly deployed at the network edge, such as the base stations (BSs) and the access points (APs), which can provide cloud-like computing service for the mobile devices. Since the mobile devices are in proximity to BSs and APs, the mobile devices are able to offload their computation-intensive tasks with high rates as well as low latency \cite{Barbarossa2014,Mach2017,Mao2017,You2017,Liu2018,Bi2018,Liang2019}.

 However, the future wireless networks are expected to serve massive devices. It may not be feasible if all data are offloaded to the AP since the AP's computation capacities are usually limited. To deal with this problem, an effective way is to explore the devices' computation resources in the network, which not only alleviates the AP's workloads but also fully utilizes the network resources. This paradigm is known as \emph{cooperative computing} \cite{Wang2016,Pu2016,You2018}. For instance, an incentive scheme was proposed in \cite{Wang2016} for encouraging mobile devices to share unused resources. A device-to-device (D2D) communication based computing collaboration and incentive mechanism was designed in \cite{Pu2016}. The authors in \cite{You2018} considered a two-user case where one user is allowed to offload computational input-data to another user.

 In order to support task offloading among multiple users in MEC, wireless communication resource  blocks (e.g., time/frequency/code) need to be utilized efficiently. Therefore, designing a suitable multiple-access scheme in MEC system is one of the most important aspects. In general, multiple-access techniques can be categorized into two different approaches \cite{Wang2006}, namely, orthogonal multiple-access (OMA) and non-orthogonal multiple-access (NOMA). In OMA schemes, users' signals are orthogonal to each other, which however cannot fully explore the capacity of the multiple-access channel and thus fundamentally limits the performance of task offloading in MEC systems. On the other hand, NOMA has been introduced as a promising solution in 5G communication systems \cite{Dai2015,Ding2017a,Ding2017b}. Different from the OMA schemes, one resource block is allowed to be allocated to multiple users in NOMA, for improving spectral efficiency and massive connections. There are a handful of works considering NOMA-aided MEC \cite{Wang2017,Wang2018a,Kiani2018,Ding2018a,Ai2017,Wu2018}.  In \cite{Ai2017} and \cite{Wu2018}, the authors considered NOMA-MEC for latency minimization.   However, these NOMA-MEC works focused on task offloading between devices and AP, without consideration of cooperative computing. In \cite{Cao2019a}, the authors considered a basic three-nodes system consisting of a user, a helper, and an AP, where the helper acts as a relay to help the user offload part of its task to the AP. Note that the helper does not have it own task to process and thus acts as a pure relay node. Moreover, the system model in \cite{Cao2019a} is based on orthogonal offloading. 

In this paper, we consider a basic three-node NOMA-aided MEC system, consisting of a user, a helper, and an AP integrated with a MEC server, as shown in Fig. \ref{fig:system}. Both the user and helper have individual computation tasks to successfully complete under a common latency constraint. The process of the proposed NOMA-aided cooperative computing is described as follows. At the first slot, the user adopts NOMA transmission to simultaneously offload its computation input-data to the helper and the AP. Then at the second slot, the helper offloads a part of its own computation input-data to the AP, in parallel with the task execution of both itself and the user.

The core idea of the proposed cooperative computing scheme is that, for a common latency constraint, at the user-side, the user has extra computation resource offered by the helper in addition to the AP. At the helper-side, if the helper executes part of the user's task, the required time of task offloading at the user can be reduced. Thus the helper in return has longer transmission time to offload its own task. Moreover, as NOMA is adopted for offloading at the user, the helper and the AP can receive and compute the user's offloaded task at the same time over the whole resource block. This further improves the offloading performance.
The main contribution of this paper is summarized as follows.
\begin{itemize}
	\item  We propose a MEC framework for cooperative computing based on NOMA transmission, in which the communication and computation resources of the network nodes are jointly designed for improving the MEC performance. Both energy consumption minimization and offloading data maximization problems are studied.
	\item For the energy consumption minimization problem, we jointly optimize the task partition, transmit power, and offloading time of both the user and helper. This problem is non-convex and we convert the original problem into a convex optimization problem which can be solved optimally. 
	\item For the offloading data maximization problem, we maximize the sum data offloaded by the user and helper by joint task partition, transmit power, and time allocation. This problem is also a non-convex problem. By analyzing the objective function, we simplify this problem with a single variable without loss in optimality. Then we propose an efficient algorithm to solve this problem globally optimally.
 \end{itemize}

The remainder of this paper is organized as follows. Section \ref{se1} introduces the system model and respectively presents the formulation of the energy consumption minimization problem and the offloading data maximization problem. The corresponding resource allocation policies for these two problems are proposed in Sections \ref{se2} and \ref{se3}, respectively. Simulation results and discussions are provided in Section \ref{se4}. Finally, Section \ref{se5} concludes the paper.

\section{System Model and Problem Formulation}\label{se1}
In this section, we first describe the system model of the considered NOMA-aided cooperative computing system. Then we formulate two optimization problems, i.e., the energy consumption minimization and the offloading data maximization problems.
\subsection{System Model}
\begin{figure}[t]
\begin{centering}
\includegraphics[scale=0.55]{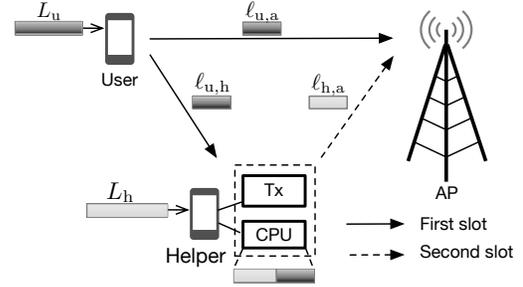}
\vspace{-0.1cm}
 \caption{System model of NOMA-aided MEC. }\label{fig:system}
\end{centering}
\vspace{-0.1cm}
\end{figure}
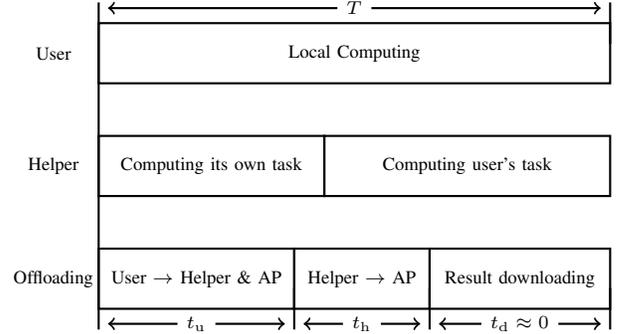
\begin{figure}[t]
  \centering
  \captionsetup{font=scriptsize}
  \begin{tikzpicture}[node distance=1.5cm, remember picture]
    \tikzstyle{state_user}=[shape=rectangle, thick, draw, minimum width=6.8cm, minimum height = 0.8cm,outer sep=0pt]
    \tikzstyle{state_helper_1}=[shape=rectangle, thick, draw, minimum width=3cm, minimum height = 0.8cm,outer sep=0pt]
    \tikzstyle{state_helper_2}=[shape=rectangle, thick, draw, minimum width=3.8cm, minimum height = 0.8cm,outer sep=0pt]
    \tikzstyle{state_offloading_1}=[shape=rectangle, thick, draw, minimum width=2.6cm, minimum height = 0.8cm,outer sep=0pt]
    \tikzstyle{state_offloading_2}=[shape=rectangle, thick, draw, minimum width=1.8cm, minimum height = 0.8cm,outer sep=0pt]
    \tikzstyle{state_offloading_3}=[shape=rectangle, thick, draw, minimum width=2.4cm, minimum height = 0.8cm,outer sep=0pt]
    \tikzset{dot/.style={circle,fill=#1,inner sep=0,minimum size=0pt}}
    \node[] (f1) at (0,0.8) [] {}; 
    \node[] (f2) at (0,-3.8) [] {}; 
    \node[] (f11) at (0,-3.6) [] {}; 
    \node[] (f12) at (0,0.6) [] {}; 
    \node[] (user) at (-0.6,0) [] {\scriptsize User}; 
    \node[state_user] (U_LC) at (3.4,0) [] {\scriptsize Local Computing};
    \node[] (user) at (-0.6,-1.5) [] {\scriptsize Helper}; 
    \node[state_helper_1] (U_LC) at (1.5,-1.5) [] {\scriptsize Computing its own task};
    \node[state_helper_2] (U_CU) at (4.9,-1.5) [] {\scriptsize Computing user's task};
    \node[] (user) at (-0.6,-3) [] {\scriptsize Offloading}; 
    \node[state_offloading_1] (U_LC) at (1.3,-3) [] {\scriptsize User $\rightarrow$ Helper \& AP};
    \node[state_offloading_2] (U_CU) at (3.5,-3) [] {\scriptsize Helper $\rightarrow$ AP};
    \node[state_offloading_3] (U_CU) at (5.6,-3) [] {\scriptsize Result downloading};
    \node[] (f3) at (3.4,0.6) [] {\scriptsize $T$};
    \node[] (f4) at (6.8,0.6) [] {}; 
    \node[] (f5) at (1.3,-3.6) [] {\scriptsize $t_{\mathrm{u}}$};
    \node[] (f6) at (2.6,-3.6) [] {}; 
    \node[] (f7) at (3.5,-3.6) [] {\scriptsize $t_{\mathrm{h}}$}; 
    \node[] (f8) at (4.4,-3.6) [] {}; 
    \node[] (f9) at (5.6,-3.6) [] {\scriptsize $t_{\mathrm{d}}\approx 0$}; 
    \node[] (f10) at (6.8,-3.6) [] {}; 
    \draw[->, thick]
    (f3) edge (f12)
    (f3) edge (f4) 
    (f5) edge (f11)
    (f5) edge (f6)
    (f7) edge (f6)
    (f7) edge (f8)
    (f9) edge (f8)
    (f9) edge (f10);
    \node[] (f19) at (6.8,0) [] {}; 
    \node[] (f20) at (6.8,0.8) [] {}; 
    \node[] (f13) at (2.6,-3) [] {}; 
    \node[] (f14) at (4.4,-3) [] {}; 
    \node[] (f15) at (6.8,-3) [] {}; 
    \node[] (f16) at (2.6,-3.8) [] {}; 
    \node[] (f17) at (4.4,-3.8) [] {}; 
    \node[] (f18) at (6.8,-3.8) [] {}; 
    \draw[thick]
          (f1) -- (f2)
          (f13) -- (f16)
          (f14) -- (f17)
          (f15) -- (f18)
          (f19) -- (f20);
  \end{tikzpicture}
  \caption{ An illustration of the proposed NOMA-aided partial offloading protocol. }\label{fig:timeslot}
\end{figure}


As shown in Fig. \ref{fig:system}, we consider a fundamental three-node model consisting of a user, a helper, and an AP. The AP is integrated with MEC server to execute the computation tasks offloaded by helper and user. In general the AP itself can act as the central node for implementing optimization and processing. That is, the AP can collect the network information and then send the optimized policies to other nodes (i.e., the user and helper) to take actions. The helper is located between the user and the AP. For easy implementation, all these three nodes are equipped with single antenna and operate in half duplex mode.  Both the user and helper have their own computation tasks to execute with data-size $L_{\mathrm{u}}$ and $L_{\mathrm{h}}$ (in bits), respectively, which should be completed within a common time $T$. In this paper, partial offloading is assumed. That is, the computation tasks of both the user and helper can be divided into  independent parts, which can be executed in parallel by local computing and offloading. Since the AP is usually connected to power grids, the energy consumption for computing at the AP can be neglected, compared with the resource-constrained mobile devices, i.e., the user and helper. Thus we only focus on the energy and latency of the user and helper in this paper.

\subsubsection{Offloading model}
 The NOMA-aided partial offloading is shown in Fig. \ref{fig:timeslot}, where the total time duration $T\in \mathbb{R}_{\geq 0}$ is divided into three slots. In the first slot $t_{\mathrm{u}}$, the user concurrently offloads parts of its computation task with $\ell_{\mathrm{u,h}}\in \mathbb{R}_{\geq 0}$ and $\ell_{\mathrm{u,a}}\in \mathbb{R}_{\geq 0}$ bits to the helper and the AP respectively using NOMA. At the second time slot $t_{\mathrm{h}}\in \mathbb{R}_{\geq 0}$, the helper partially offloads $\ell_{\mathrm{h,a}}\in \mathbb{R}_{\geq 0}$ bits to the AP. In general, the user and helper can offload data simultaneously in the second time slot \cite{Liu2019}. Our framework and algorithms are applicable for this case as well. However, to simply the frame structure and make the analysis more tractable, we consider that only the helper offloads data in the second time slot. The third slot $t_{\mathrm{d}}\in \mathbb{R}_{\geq 0}$ is used for downloading the computation results. To meet the latency constraint of the user and helper, we have $t_{\mathrm{u}}+t_{\mathrm{h}}+t_{\mathrm{d}}\leq T$ in the offloading process. Generally, the computed results are in small sizes and can be transmitted for a short time.  For simplicity, we assume that the downloading time for the user and helper is negligible, i.e., $t_{\mathrm{d}} \approx 0$. As a result, the latency constraint can be further simplified as $t_{\mathrm{u}}+t_{\mathrm{h}}\leq T$.

Specifically, in the first time slot $t_{\mathrm{u}}$, i.e., the NOMA transmission process, the user transmits a linear superposition of the data to the helper and the AP by allocating transmit power $p_{\mathrm{u,h}}\in \mathbb{R}_{\geq 0}$ and $p_{\mathrm{u,a}}\in \mathbb{R}_{\geq 0}$, respectively. The received signals at the helper and the AP are given by
\begin{align}
	y^{(1)}_{\mathrm{u,h}}&=\sqrt{p_{\mathrm{u,h}}}x_{\mathrm{u,h}}g_{\mathrm{u,h}}+\sqrt{p_{\mathrm{u,a}}}x_{\mathrm{u,a}}g_{\mathrm{u,h}}+z_{\mathrm{h}},\label{eqn:r2}\\
	y^{(1)}_{\mathrm{u,a}}&=\sqrt{p_{\mathrm{u,h}}}x_{\mathrm{u,h}}g_{\mathrm{u,a}}+\sqrt{p_{\mathrm{u,a}}}x_{\mathrm{u,a}}g_{\mathrm{u,a}}+z_{\mathrm{a}},\label{eqn:r3}
\end{align}
where $x_{\mathrm{u,h}}\in \mathbb{C}$ and $x_{\mathrm{u,a}}\in \mathbb{C}$ are the transmitted signal for the helper and the AP, respectively; $g_{\mathrm{u,h}}\in \mathbb{C}$ and $g_{\mathrm{u,a}}\in \mathbb{C}$ denote the channel coefficients from the user to the helper and the AP, respectively. It is assumed that the wireless links experience independent and identically distributed (i.i.d.) block Rayleigh fading. Constants $z_{\mathrm{h}}\in \mathbb{C}$ and $z_{\mathrm{a}}\in \mathbb{C}$ indicate the additive white Gaussian noise (AWGN) at the helper and AP respectively with zero mean and variance $\sigma^{2}\in \mathbb{R}_{\geq 0}$. We suppose that global channel state information (CSI) is available. At the receivers, the helper first decodes the signal of the AP and subtracts it from the received signal using successive interference cancellation (SIC) when decoding its own signal.

Then the achievable rates from the user to the helper and the AP can be respectively expressed as
\begin{align}
	R^{(1)}_{\mathrm{u,h}}&=\log_{2}\left(1+h_{\mathrm{u,h}}p_{\mathrm{u,h}}\right),\label{eqn:r5}\\
	R^{(1)}_{\mathrm{u,a}}&=\log_{2}\left(1+\frac{h_{\mathrm{u,a}}p_{\mathrm{u,a}}}{1+h_{\mathrm{u,a}}p_{\mathrm{u,h}}}\right),\label{eqn:r6}
\end{align}
where $h_{\mathrm{u,h}}:=|g_{\mathrm{u,h}}|^{2}/\sigma^{2}$ and $h_{\mathrm{u,a}}:=|g_{\mathrm{u,a}}|^{2}/\sigma^{2}$ are defined to be the effective channel gain to noise power ratio (CGNR) from the user to the helper and the AP, respectively. 

The corresponding offloaded bits from the user to the helper and the AP are respectively defined to be
\begin{eqnarray}
	&\ell_{\mathrm{u,h}}:=t_{\mathrm{u}}R^{(1)}_{\mathrm{u,h}},\label{eqref:r66} \\
	&\ell_{\mathrm{u,a}}:=t_{\mathrm{u}}R^{(1)}_{\mathrm{u,a}}.\label{eqref:r67}
\end{eqnarray}
The energy consumption of the NOMA-aided offloading over duration time $t_{\mathrm{u}}$ at the user is defined to be
\begin{align}
	E_{\mathrm{u}}^{\rm{off}}:=t_{\mathrm{u}}(p_{\mathrm{u,h}}+p_{\mathrm{u,a}}).
\end{align}

At the second time slot $t_{\mathrm{h}}$, the helper offloads $\ell_{\mathrm{h,a}}\in \mathbb{R}_{\geq 0}$ bits, a part of the input task $L_{\mathrm{h}}\in \mathbb{R}_{\geq 0}$, to the AP with transmit power $p_{\mathrm{h,a}}$. Similarly, let us define $h_{\mathrm{h,a}}:=|g_{\mathrm{h,a}}|^{2}/\sigma^{2}$ to be the effective CGNR from the helper to the AP, where $g_{\mathrm{h,a}}\in \mathbb{C}$ denotes the channel coefficients from the helper to the AP. The achievable rate for task offloading from the helper to the AP is given by
\begin{align}
	R^{(2)}_{\mathrm{h,a}}&=\log_{2}\left(1+h_{\mathrm{h,a}}p_{\mathrm{h,a}}\right).\label{eqn:r7}
\end{align}
Then the offloaded data bits are defined to be
\begin{align}
\ell_{\mathrm{h,a}}:=t_{\mathrm{h}}R^{(2)}_{\mathrm{h,a}}.\label{eqn:data_ha}
\end{align}
The energy consumption of data offloading  over duration time $t_{\mathrm{h}}$ at the helper is defined to be
\begin{align}
	E^{\rm{off}}_{\mathrm{h}}:=t_{\mathrm{h}}p_{\mathrm{h,a}}. \label{eqn:r13}
\end{align}
\subsubsection{Computing model}
Denote the central processing unit (CPU) frequency of the user and helper  as $f_{\mathrm{u}}\in \mathbb{R}_{\geq 0}$ and $f_{\mathrm{h}}\in \mathbb{R}_{\geq 0}$, respectively. As mentioned before, the user has to compute $L_{\mathrm{u}}$ bits input-data, in which $\ell_{\mathrm{u,h}}$ and $\ell_{\mathrm{u,a}}$ bits are offloaded to the helper and the AP respectively, and the rest, i.e., $L_{\mathrm{u}}-\ell_{\mathrm{u,h}}-\ell_{\mathrm{u,a}}$ bits, are computed locally. The energy consumption for local computing at user is defined to be
\begin{align}
E_{\mathrm{u}}^{\rm{loc}}:=(L_{\mathrm{u}}-\ell_{\mathrm{u,h}}-\ell_{\mathrm{u,a}})\kappa f_{\mathrm{u}}^{2},\label{eqn:r8}
\end{align}
where $\kappa\in \mathbb{R}_{\geq 0}$ denotes a constant related to the hardware architecture \cite{Mao2017}. Note that the user's local computing can be performed in the whole duration. The user's local computing latency constraint is defined to be
\begin{align}
\frac{(L_{\mathrm{u}}-\ell_{\mathrm{u,h}}-\ell_{\mathrm{u,a}})C_{\mathrm{u}}}{f_{\mathrm{u}}}\leq T,\label{eqn:r9}
\end{align}
where $C_{\mathrm{u}}\in \mathbb{R}_{\geq 0}$, namely computation intensity, is the number of CPU cycles required for computing 1-bit of input data for the user.
As for the helper, its own $\ell_{\mathrm{h,a}}$ bits are offloaded to the AP at the second slot within duration $t_{\mathrm{h}}$, and $L_{\mathrm{h}}-\ell_{\mathrm{h,a}}+\ell_{\mathrm{u,h}}$ bits are left for local computing, where $\ell_{\mathrm{u,h}}$ bits are received from the user in the first slot with duration $t_{\mathrm{u}}$. Then the energy consumption for local computing at the helper is defined to be
\begin{align}
E_{\mathrm{h}}^{\rm{loc}}:=(L_{\mathrm{h}}-\ell_{\mathrm{h,a}}+\ell_{\mathrm{u,h}})\kappa f_{\mathrm{h}}^{2}.\label{eqn:r10}
\end{align}

When the user offloads task (i.e., the NOMA transmission process), the helper first computes its own $L_{\mathrm{h}}-\ell_{\mathrm{h,a}}$ bits and receives $\ell_{\mathrm{u,h}}$ bits offloaded by the user simultaneously. After computing $L_{\mathrm{h}}-\ell_{\mathrm{h,a}}$ bits, if the helper has already received $\ell_{\mathrm{u,h}}$ bits, i.e., $(L_{\mathrm{h}}-\ell_{\mathrm{h,a}})C_{\mathrm{u}}/f_{\mathrm{h}}\geq t_{\mathrm{u}}$, it turns to execute the task $\ell_{\mathrm{u,h}}$ offloaded by the user. Otherwise, it keeps receiving data from the user until $\ell_{\mathrm{u,h}}$ bits are all received. Therefore, we define the helper's latency constraint to be
\begin{align}
\max\left(t_{\mathrm{u}},\frac{(L_{\mathrm{h}}-\ell_{\mathrm{h,a}})C_{\mathrm{u}}}{f_{\mathrm{h}}}\right)+\frac{\ell_{\mathrm{u,h}}C_{\mathrm{u}}}{f_{\mathrm{h}}} \leq T.\label{eqn:r11}
\end{align}

\subsection{Problem Formulation}

\subsubsection{Energy Consumption Minimization Problem}
Here we consider the energy consumption as the system performance metric. The total energy consumption for the user and the helper are defined to be $E^{\rm{off}}_{\mathrm{u}}+E^{\rm{loc}}_{\mathrm{u}}$ and $E^{\rm{off}}_{\mathrm{h}}+E^{\rm{loc}}_{\mathrm{h}}$, respectively. With the objective of minimizing the total energy consumption of the user and helper, subjected to the common latency constraint, the optimization problem is defined to be 
\begin{align}
 {\text{(P1)}:}\min_{\bm{p},\bm{\ell},\bm{t}}&\quad w_{\mathrm{u}}\left(E^{\rm{off}}_{\mathrm{u}}
 +E^{\rm{loc}}_{\mathrm{u}}\right)
 +w_{\mathrm{h}}\left(E^{\rm{off}}_{\mathrm{h}}
 +E^{\rm{loc}}_{\mathrm{h}}\right)\nonumber \\
{\rm s.t.} &\quad\frac{(L_{\mathrm{u}}-\ell_{\mathrm{u,h}}-\ell_{\mathrm{u,a}})C_{\mathrm{u}}}{f_{\mathrm{u}}}\leq T,\label{eqn:r14}\\
  & \quad\max\left(t_{\mathrm{u}},\frac{(L_{\mathrm{h}}-\ell_{\mathrm{h,a}})C_{\mathrm{u}}}{f_{\mathrm{h}}}\right)+\frac{\ell_{\mathrm{u,h}}C_{\mathrm{u}}}{f_{\mathrm{h}}} \leq T,\label{eqn:r15}\\
  &\quad t_{\mathrm{u}}+t_{\mathrm{h}}\leq T\label{eqn:r16},\\
  &\quad\ell_{\mathrm{u,a}}C_{\mathrm{u}} + \ell_{\mathrm{h,a}}C_{\mathrm{h}} \leq F,\label{AP_com_capacity} \\
  &\quad \ell_{\mathrm{h,a}}\leq L_{\mathrm{h}},\label{eqn:r60}\\
  &\quad \bm{p}\in \mathbb{R}^{3}_{\geq 0},
  \bm{\ell}\in \mathbb{R}^{3}_{\geq 0},
  \bm{t}\in \mathbb{R}^{2}_{\geq 0},
  \end{align}
where $\bm{p}:=\{p_{\mathrm{u,h}},p_{\mathrm{u,a}},p_{\mathrm{h,a}}\}\in \mathbb{R}^{3}_{\geq 0}$, $\bm{\ell}:=\{\ell_{\mathrm{u,h}},\ell_{\mathrm{u,a}},\ell_{\mathrm{h,a}}\}\in\mathbb{R}^{3}_{\geq 0}$, and $\bm{t}:=\{t_{\mathrm{u}},t_{\mathrm{h}}\}\in \mathbb{R}^{2}_{\geq 0}$. The constants $w_{\mathrm{u}}\in \mathbb{R}_{\geq 0}$ and $w_{\mathrm{h}}\in \mathbb{R}_{\geq 0}$ are weighted factors decided by the system. Note that $w_{\mathrm{u}}$ and $w_{\mathrm{h}}$ respectively account for the priorities of the user and the helper. \eqref{eqn:r14} and \eqref{eqn:r15} denote the latency constraints of the local computing at the user and the helper, respectively. \eqref{eqn:r16} is the total offloading time constraint of the user and helper. \eqref{AP_com_capacity} denotes the computation capacity of the AP, where $C_{\mathrm{h}}$ denotes the number of CPU cycles for computing 1-bit of input-data of the helper and $F$ is the server's available computational capacity. \eqref{eqn:r60} ensures that the offloading data of the helper cannot exceed its input-data size.
\subsubsection{Offloading Data Maximization Problem}
We also consider the system's maximum sum of offloading data. The offloading data in the user and helper are $\ell_{\mathrm{u,h}}+\ell_{\mathrm{u,a}}$ and $\ell_{\mathrm{h,a}}$, respectively. Denote $\bar{P}_{\mathrm{u}}\in \mathbb{R}_{\geq 0}$ and $\bar{P}_{\mathrm{h}}\in \mathbb{R}_{\geq 0}$ as the maximum transmit power of the user and helper, respectively. Then the offloading data maximization problem is defined to be
\begin{align}
 {\text{(P2)}:}\max_{\bm{p}, \bm{\ell},\bm{t}}&\quad w_{\mathrm{u}}(\ell_{\mathrm{u,h}}+\ell_{\mathrm{u,a}})+w_{\mathrm{h}}\ell_{\mathrm{h,a}}\nonumber \\
{\rm s.t.} &\quad p_{\mathrm{u,h}}+p_{\mathrm{u,a}}\leq\bar{P}_{\mathrm{u}},\label{eqn:r17}\\
  &\quad  p_{\mathrm{h,a}}\leq \bar{P}_{\mathrm{h}},\label{max_trasmit_power_helper}\\
  &\quad  \kappa f^{2}_{\mathrm{h}} \ell_{\mathrm{u,h}} \leq E'_{\mathrm{h}},\label{cons_available_energy_helper}
\\
  &\quad  t_{\mathrm{u}}+\frac{\ell_{\mathrm{u,h}}C_{\mathrm{u}}}{f_{\mathrm{h}}} \leq T,\label{eqn:r18}\\
  &\quad  t_{\mathrm{u}}+t_{\mathrm{h}}\leq T,\label{eqn:r19}\\
  &\quad \bm{p}\in \mathbb{R}^{3}_{\geq 0},
  \bm{\ell}\in \mathbb{R}^{3}_{\geq 0},
  \bm{t}\in \mathbb{R}^{2}_{\geq 0},
\end{align}
where \eqref{eqn:r17} and \eqref{max_trasmit_power_helper} denote the maximum transmit power constraints of the user and the helper, respectively. \eqref{cons_available_energy_helper} denotes the limit of available energy at the helper for processing the offloaded data, where $E'_{\mathrm{h}}\in \mathbb{R}_{\geq 0}$ is a fixed number denoting helper's available energy for processing user's offloading data. \eqref{eqn:r18} is the local computing latency constraint of the helper, and \eqref{eqn:r19} is the total offloading time constraint of the user and helper. Note that the issue of local computing at the user and helper is not necessarily considered because the goal of Problem (P2) is to maximize the offloading data. Intuitively, $p_{\mathrm{h,a}}=\bar{P}_{\mathrm{h}}$ holds for offloading data maximization. 

\section{Optimal Solution for Energy Consumption Minimization Problem}\label{se2}
In this section, we address the energy consumption minimization Problem (P1). As the problem is non-convex, we first transform it into a convex problem by some mathematical methods. Then we develop an efficient algorithm to obtain the globally optimal solution.
\subsection{Problem Transformation}
With  \eqref{eqref:r66} and \eqref{eqn:data_ha}, we can rewrite $p_{\mathrm{u,h}}$ and $p_{\mathrm{h,a}}$ as
\begin{align}
p_{\mathrm{u,h}}=\frac{1}{h_{\mathrm{u,h}}}f_{1}\left(\frac{\ell_{\mathrm{u,h}}}{t_{\mathrm{u}}}\right),\label{eqn:r20}\\
p_{\mathrm{h,a}}=\frac{1}{h_{\mathrm{h,a}}}f_{1}\left(\frac{\ell_{\mathrm{h,a}}}{t_{\mathrm{h}}}\right),\label{eqn:r21}
\end{align}
where $f_{1}$ represents the mapping $f_{1}:\mathbb{R}\rightarrow \mathbb{R},  x\mapsto 2^{x}-1$ for all $x\in \mathbb{R}$. Also with \eqref{eqref:r67}, the transmit power $p_{\mathrm{u,a}}$ can be rewritten as
\begin{align}
p_{\mathrm{u,a}}=&
\frac{1}{h_{\mathrm{u,h}}}f_{1}\left(\frac{\ell_{\mathrm{u,a}}
+\ell_{\mathrm{u,h}}}{t_{\mathrm{u}}}\right)
-\frac{1}{h_{\mathrm{u,h}}}f_{1}
\left(\frac{\ell_{\mathrm{u,h}}}
{t_{\mathrm{u}}}\right)\nonumber\\&
+\left(\frac{1}{h_{\mathrm{u,a}}}
-\frac{1}{h_{\mathrm{u,h}}}\right)
f_{1}\left(\frac{\ell_{\mathrm{u,a}}}
{t_{\mathrm{u}}}\right).\label{eqn:r22}
\end{align}
 Then we find that objective function of Problem (P1) equals $f_{2}$, which is defined to be the mapping
\begin{align}
f_{2}:\mathbb{R}^{3}_{\geq 0}\times 
\mathbb{R}^{2}_{\geq 0}\rightarrow\ & 
\mathbb{R},\nonumber\\ 
(\bm{\ell},\bm{t})
\mapsto\
&\frac{w_{\mathrm{u}}t_{\mathrm{u}}}
{h_{\mathrm{u,h}}}f_{1}\left( \frac{\ell_{\mathrm{u,h}} + \ell_{\mathrm{u,a}}}
{t_{\mathrm{u}}} \right)\nonumber \\&
+w_{\mathrm{u}}\kappa f_{\mathrm{u}}^{2}(L_{\mathrm{u}}-\ell_{\mathrm{u,h}}-\ell_{\mathrm{u,a}})
\nonumber\\&+\left(\frac{w_{\mathrm{u}}t_{\mathrm{u}}}{h_{\mathrm{u,a}}}-\frac{w_{\mathrm{u}}t_{\mathrm{u}}}{h_{\mathrm{u,h}}}\right)f_{1}\left(\frac{\ell_{\mathrm{u,a}}}{t_{\mathrm{u}}}\right)
\nonumber \\&+\frac{w_{\mathrm{h}}t_{\mathrm{h}}}{h_{\mathrm{h,a}}}f_{1}\left(\frac{\ell_{\mathrm{h,a}}}{t_{\mathrm{h}}}\right)\nonumber\\&+w_{\mathrm{h}}\kappa f_{\mathrm{h}}^{2}(L_{\mathrm{h}}-\ell_{\mathrm{h,a}}+\ell_{\mathrm{u,h}}),\label{eqn:r23}
\end{align}
where $\mathbb{R}^{3}_{\geq 0}\times \mathbb{R}^{2}_{\geq 0}$ is the Cartesian product of the sets $\mathbb{R}^{3}_{\geq 0}$ and $\mathbb{R}^{2}_{\geq 0}$.
\begin{lemma}\label{Lm5}
Suppose $h_{\mathrm{u,a}}\leq h_{\mathrm{u,h}}, f_{2}(\bm{\ell},\bm{t})$ is jointly convex with respect to $\bm{\ell}$ and $\bm{t}$ over Problem (P1)'s feasible set.
\end{lemma}

\begin{proof}
	Please refer to Appendix \ref{AP5}.
\end{proof}

With Lemma \ref{Lm5}, we can further prove that Problem (P1) is a convex problem and thus can be solved by convex optimization methods.

\subsection{Finding Optimal Solution}
The partial derivative of $f_{2}(\bm{\ell},\bm{t})$ with respect to $t_{\mathrm{h}}$ is
\begin{align}
	\frac{\partial f_{2}}{\partial t_{\mathrm{h}}}=\frac{w_{\mathrm{h}}}{h_{\mathrm{h,a}}}\left(\left(1-\frac{\ell_{\mathrm{h,a}}\ln2}{t_{\mathrm{h}}}\right) 2^{\frac{\ell_{\mathrm{h,a}}}{t_{\mathrm{h}}}}-1\right).
\end{align}
Then we introduce the following lemma to further simplify Problem (P1).
\begin{lemma}\label{Lm1}
Inequality $\bigl(1-x\ln2\bigr) 2^{x}-1\leq 0$ is always satisfied for $x\geq 0$.
\end{lemma}

\begin{proof}
	Please refer to Appendix \ref{AP1}.
\end{proof}

With Lemma \ref{Lm1}, we can draw a conclusion that $\frac{\partial f_{2}}{\partial t_{\mathrm{h}}}\leq 0$ is satisfied when $\ell_{\mathrm{h,a}}/t_{\mathrm{h}}\geq 0$, i.e., $\ell_{\mathrm{h,a}}\geq 0$ and $t_{\mathrm{h}}\geq 0$. As a result, to minimize the value of the objective function $f_{2}(\bm{\ell},\bm{t})$, constraint \eqref{eqn:r16} needs to be satisfied with equality, which means that the whole time has to be fully utilized. Otherwise, we can still further improve the value of $f_{2}$ by increasing $t_{\mathrm{h}}$ until the equality is activated.  This property is consistent with the intuition that with fixed offloading bits $\ell_{\mathrm{h,a}}$, the longer transmission time $t_{\mathrm{h}}$ the helper uses, the less offloading energy $E_{\mathrm{h}}^{\rm{off}}$ it consumes.

As a result, we can represent $t_{\mathrm{u}}$ with $\alpha T$ and $t_{\mathrm{h}}$ with $(1-\alpha)T$ based on Lemma \ref{Lm1}, where $\alpha\in \mathbb{R}_{\geq 0}$ denotes the proportion of the user's transmission time in the whole time $T$. Here we solve Problem (P1) with fixed $\alpha$ at the first step and then obtain the optimal value of $\alpha$ by numerical search within $0\leq\alpha\leq1$.

Note that the constraint \eqref{eqn:r15} in Problem (P1) is equivalent to the constraints $\alpha T+\frac{\ell_{\mathrm{u,h}}C_{\mathrm{u}}}{f_{\mathrm{h}}}\leq T$ and $\frac{(L_{\mathrm{h}}-\ell_{\mathrm{h,a}}+\ell_{\mathrm{u,h}})C_{\mathrm{u}}}{f_{\mathrm{h}}}\leq T$. Therefore, given $\alpha$, Problem (P1) can be reformulated as
\begin{eqnarray}
{\text{ (P1')}:} \min_{\bm{\ell}\geq \bm{0}}&&f_{2}(\bm{\ell})\nonumber \\
{\rm s.t.}&&(L_{\mathrm{u}}-\ell_{\mathrm{u,h}}-\ell_{\mathrm{u,a}})C_{\mathrm{u}}\leq f_{\mathrm{u}}T,\label{eqn:r25}\\
  && (L_{\mathrm{h}}-\ell_{\mathrm{h,a}}+\ell_{\mathrm{u,h}})C_{\mathrm{u}}\leq f_{\mathrm{h}}T,\label{eqn:r26}\\
   && \alpha T+\frac{\ell_{\mathrm{u,h}}C_{\mathrm{u}}}{f_{\mathrm{h}}} \leq T\label{eqn:r27}, \\
   &&\ell_{\mathrm{u,a}}C_{\mathrm{u}} + \ell_{\mathrm{h,a}}C_{\mathrm{h}} \leq F,\\
    && \ell_{\mathrm{h,a}}\leq L_{\mathrm{h}},\label{eqn:r61}
  \end{eqnarray}
where $f_{2}(\bm{\ell})$ can be obtained by setting $t_{\mathrm{u}}=\alpha T$ and $t_{\mathrm{h}}=(1-\alpha)T$ in $f_{2}(\bm{\ell},\bm{t})$. We construct the partial Lagrangian function of the form
\begin{align}
L_1(\bm{\ell}, \bm{\lambda})=&\frac{w_{\mathrm{u}}\alpha T}{h_{\mathrm{u,h}}}f_{1}\left(\frac{\ell_{\mathrm{u,h}}+\ell_{\mathrm{u,a}}}{\alpha T}\right)
+\frac{w_{\mathrm{h}}(1-\alpha) T}{h_{\mathrm{h,a}}}\nonumber\\&\cdot f_{1}\left(\frac{\ell_{\mathrm{h,a}}}{(1-\alpha)T}\right)
+w_{\mathrm{u}}\kappa f_{\mathrm{u}}^{2}(L_{\mathrm{u}}-\ell_{\mathrm{u,h}}-\ell_{\mathrm{u,a}})\nonumber\\&+\left(\frac{w_{\mathrm{u}}\alpha T}{h_{\mathrm{u,a}}}
-\frac{w_{\mathrm{u}}\alpha T}{h_{\mathrm{u,h}}}\right)f_{1}\left(\frac{\ell_{\mathrm{u,a}}}{\alpha T}\right)
\nonumber\\&+w_{\mathrm{h}}\kappa f_{\mathrm{h}}^{2}(L_{\mathrm{h}}-\ell_{\mathrm{h,a}}+\ell_{\mathrm{u,h}})
\nonumber\\&+\lambda_{1}((L_{\mathrm{u}}-\ell_{\mathrm{u,h}}-\ell_{\mathrm{u,a}})C_{\mathrm{u}}-f_{\mathrm{u}}T)
\nonumber\\&+\lambda_{2}((L_{\mathrm{h}}-\ell_{\mathrm{h,a}}+\ell_{\mathrm{u,h}})C_{\mathrm{u}}-f_{\mathrm{h}}T)
\nonumber\\&+\lambda_{3}(\ell_{\mathrm{u,a}}C_{\mathrm{u}} + \ell_{\mathrm{h,a}}C_{\mathrm{h}} -F),\label{eqn:r28}
\end{align}
where $\bm{\lambda}:=\{\lambda_{1},\lambda_{2},\lambda_{3}\}\in \mathbb{R}^{3}_{\geq 0}$ is the collection of the Lagrangian multipliers associated with constraints \eqref{eqn:r25} and \eqref{eqn:r26}, respectively. Then the dual function is defined to be
\begin{eqnarray}\label{eqn:r29}
g_1(\bm{\lambda}):=\begin{cases}
\min_{\bm{\ell}}\quad L_1(\bm{\ell},\bm{\lambda})\nonumber\\
{\rm s.t.}\quad
\alpha T+\frac{\ell_{\mathrm{u,h}}C_{\mathrm{u}}}{f_{\mathrm{h}}} \leq T,
    \ell_{\mathrm{h,a}}\leq L_{\mathrm{h}},
    \bm{\ell}\geq \bm{0} .
\end{cases}
\end{eqnarray}
Thereby the dual problem is 
\begin{align}\label{eqn:r30}
\max_{\bm{\lambda}\geq\bm{0} }g_1(\bm{\lambda}).
\end{align}

Since Problem (P1') is convex with given $\alpha$ and satisfies the Slater's condition, there is zero duality gap between Problem (P1') and problem \eqref{eqn:r30}. In the following, we firstly solve problem \eqref{eqn:r29} and then update the Lagrangian multipliers $\bm{\lambda}$ until they converge to optimal values. Then we obtain the optimal solution of Problem (P1') which is denoted by the vector $\bm{\ell}^*:=\{\ell_{\mathrm{u,h}}^{*},\ell_{\mathrm{u,a}}^{*},\ell_{\mathrm{h,a}}^{*}\}$.

For given Lagrangian variables $\bm{\lambda}$, the partial derivative of $L_1(\bm{\ell}, \bm{\lambda})$ with respect to $\ell_{\mathrm{u,h}}$ is
\begin{align}\label{eqn:r31}
\frac{\partial L_1}{\partial \ell_{\mathrm{u,h}}}=&\frac{w_{\mathrm{u}}\ln2}{h_{\mathrm{u,h}}}2^{\frac{\ell_{\mathrm{u,h}}+\ell_{\mathrm{u,a}}}{\alpha T}} - w_{\mathrm{u}}\kappa f_{\mathrm{u}}^{2} + w_{\mathrm{h}}\kappa f_{\mathrm{h}}^{2}-\lambda_{1}C_{\mathrm{u}} \nonumber\\&+ \lambda_{2}C_{\mathrm{u}}.
\end{align}
The optimal value of $\ell_{\mathrm{u,h}}+\ell_{\mathrm{u,a}}$, can be obtained by solving $\frac{\partial L_1}{\partial \ell_{\mathrm{u,h}}}=0$. We have
\begin{align}\label{eqn:r32}
\ell_{\mathrm{u,h}}^*+\ell_{\mathrm{u,a}}^*=\left\{\begin{array}{ll}
0&A \leq 1,\\
\alpha T\log_{2}A
&
\text{otherwise}.
\end{array}\right.
\end{align}
Here $A:=(w_{\mathrm{u}}\kappa f_{\mathrm{u}}^{2}-w_{\mathrm{h}}\kappa f_{\mathrm{h}}^{2}+\lambda_{1}C_{\mathrm{u}}-\lambda_{2}C_{\mathrm{u}})h_{\mathrm{u,h}}/(w_{\mathrm{u}}\ln2)\in \mathbb{R}$. The amount of the data offloaded by the user increases with $w_{\mathrm{u}}\kappa f_{\mathrm{u}}^{2}-w_{\mathrm{h}}\kappa f_{\mathrm{h}}^{2}$. Note that because of the nonnegative constraints $\ell_{\mathrm{u,h}}\geq0$ and $\ell_{\mathrm{u,a}}\geq0$, when $A\leq1$, i.e., $\ell_{\mathrm{u,h}}^*+\ell_{\mathrm{u,a}}^*=0$, we have $\ell_{\mathrm{u,h}}^*=0$ and $\ell_{\mathrm{u,a}}^*=0$. On the other hand, in terms of the case $A>1$, we need to solve $\ell_{\mathrm{u,h}}^*$ by computing the partial derivative of $L_1$ with respect to $\ell_{\mathrm{u,a}}$. We have
\begin{align}\label{eqn:r33}
\frac{\partial L_1}{\partial \ell_{\mathrm{u,a}}}=&
w_{\mathrm{u}}2^{\frac{\ell_{\mathrm{u,a}}}{\alpha T}}\ln2
\left( \frac{1}{h_{\mathrm{u,h}}}2^{\frac{\ell_{\mathrm{u,h}}}{\alpha T}}
+\frac{1}{h_{\mathrm{u,a}}} 
- \frac{1}{h_{\mathrm{u,h}}} \right)\nonumber\\&
-w_{\mathrm{u}}\kappa f_{\mathrm{u}}^{2}+(\lambda_{3}-\lambda_{1})C_{\mathrm{u}}.
\end{align}
With $\ell_{\mathrm{u,h}}^*+\ell_{\mathrm{u,a}}^*=\alpha T\log_{2}A$, the partial derivative $\frac{\partial L_1}{\partial \ell_{\mathrm{u,a}}}$ can be further simplified as
\begin{align}\label{eqn:r34}
\frac{\partial L_1}{\partial \ell_{\mathrm{u,a}}}=&w_{\mathrm{u}}2^{\frac{\ell_{\mathrm{u,a}}}{\alpha T}}
\ln2\left(\frac{1}{h_{\mathrm{u,a}}}-\frac{1}{h_{\mathrm{u,h}}}\right)
-w_{\mathrm{h}}\kappa f_{\mathrm{h}}^{2}\nonumber\\&+(\lambda_{3}-\lambda_{2})C_{\mathrm{u}}.
\end{align}
Then we can solve $\ell_{\mathrm{u,a}}^*$ by setting  $\frac{\partial L_1}{\partial \ell_{\mathrm{u,a}}}$ equal to zero as well as considering the nonnegative constraint $\ell_{\mathrm{u,a}}\geq0$. One can verify that $\ell_{\mathrm{u,a}}^*$ equals
\begin{align}\label{eqn:r35}
\ell_{\mathrm{u,a}}^*=\left\{\begin{array}{ll}
0& A \leq 1,\\
\alpha T\bigl[\log_{2}C\bigr]^+
&
\text{otherwise},
\end{array}\right.
\end{align}
where $C:=(w_{\mathrm{h}}\kappa f_{\mathrm{h}}^{2}+(\lambda_{2}-\lambda_{3})C_{\mathrm{u}})h_{\mathrm{u,h}}h_{\mathrm{u,a}}/(w_{\mathrm{u}}\ln2(h_{\mathrm{u,h}}-h_{\mathrm{u,a}}))\in \mathbb{R}$ and $[x]^{+}:=\max(x,0)$. The result \eqref{eqn:r35} shows that when $w_{\mathrm{h}}\kappa f_{\mathrm{h}}^{2}$ is sufficiently large, which means that the helper is more energy-consuming for local computing, the user prefers to offload more bits to the AP rather than the helper.

Subtracting $\ell_{\mathrm{u,a}}^*$ from $\ell_{\mathrm{u,h}}^*+\ell_{\mathrm{u,a}}^*$ as well as considering the constraint \eqref{eqn:r27}, we have
\begin{align}\label{eqn:r36}
\ell_{\mathrm{u,h}}^*=\left\{\begin{array}{ll}
\min(\alpha T\bigl[\log_{2}A\bigr]^{+},(1-\alpha)Tf_{\mathrm{h}})&C \leq 1,\\
\min(\alpha T\bigl[\log_{2}\frac{A}{C}\bigr]^{+},(1-\alpha)Tf_{\mathrm{h}})
&
\text{otherwise}.
\end{array}\right.
\end{align}
Note that $A/C=(w_{\mathrm{u}}\kappa f_{\mathrm{u}}^{2}-w_{\mathrm{h}}\kappa f_{\mathrm{h}}^{2}+\lambda_{1}C_{\mathrm{u}}-\lambda_{2}C_{\mathrm{u}})(h_{\mathrm{u,h}}-h_{\mathrm{u,a}}))/(h_{\mathrm{u,a}}(w_{\mathrm{h}}\kappa f_{\mathrm{h}}^{2}+(\lambda_{2}-\lambda_{3})C_{\mathrm{u}}))$. We observe that $\ell_{\mathrm{u,h}}^*$ is an increasing function with respect to $w_{\mathrm{u}}\kappa f_{\mathrm{u}}^{2}$, which implies that, in order to reduce the energy consumption, the user needs to offload more bits to the helper if its per-bit energy consumption of local computing is higher than that of the helper. Otherwise the user prefers to locally compute more bits. Furthermore, from \eqref{eqn:r35} and \eqref{eqn:r36}, it is shown that when the gap between $h_{\mathrm{u,a}}$ and $h_{\mathrm{u,h}}$ widens, $\ell_{\mathrm{u,a}}$ declines and $\ell_{\mathrm{u,h}}$ increases.  This is consistent with the intuition that more bits are offloaded in the good channel.

We can obtain
$\ell_{\mathrm{h,a}}^*$ by setting $\frac{\partial L_1}{\partial \ell_{\mathrm{h,a}}}$ equal to zero as well as considering the nonnegative constraint $\ell_{\mathrm{h,a}}\geq 0$ and constraint \eqref{eqn:r61}, which is given by
\begin{align}\label{eqn:r38}
\ell_{\mathrm{h,a}}^*=&\min\biggl(L_{\mathrm{h}}, \nonumber\\&(1-\alpha)T\left(\log_{2}
\frac{h_{\mathrm{h,a}}(w_{\mathrm{h}}\kappa f_{\mathrm{h}}^{2}
+\lambda_{2}C_{\mathrm{u}}-\lambda_{3}C_{\mathrm{h}})}
{w_{\mathrm{u}}\ln2}\right)^{+}\biggr).
\end{align}

After solving problem \eqref{eqn:r29} with given $\bm{\lambda}$, we obtain $\bm{\lambda}^*$ which is the solution of the maximization problem \eqref{eqn:r30}. Because of the convexity of problem \eqref{eqn:r30}, we adopt the ellipsoid method to update $\bm{\lambda}$ until the elements of $\bm{\lambda}$ converge the optimal values. The subgradients used for the ellipsoid method are provided as
\begin{align}
&\Delta \lambda_{1}= (L_{\mathrm{u}}-\ell_{\mathrm{u,h}}^*-\ell_{\mathrm{u,a}}^*)C_{\mathrm{u}}-f_{\mathrm{u}}T,\label{eqn:r39}\\
&\Delta \lambda_{2}=(L_{\mathrm{h}}-\ell_{\mathrm{h,a}}^*+\ell_{\mathrm{u,h}}^*)C_{\mathrm{u}}-f_{\mathrm{h}}T,\label{eqn:r40}\\
&\Delta \lambda_{3}=\ell_{\mathrm{u,a}}^{*}C_{\mathrm{u}} + \ell_{\mathrm{h,a}}^{*}C_{\mathrm{h}} - F,\label{eqn:subgradiant_lam3}
\end{align}
where $\ell_{\mathrm{u,a}}^*$, $\ell_{\mathrm{u,h}}^*$, and $\ell_{\mathrm{h,a}}^*$ are solved in \eqref{eqn:r35}, \eqref{eqn:r36}, and  \eqref{eqn:r38}, respectively. Because of the zero duality gap, the solution of Problem (P1') comes out when $\bm{\lambda^*}$ is achieved.

Finally, we determine the optimal solution of $\alpha$, denoted as $\alpha^{*}$, which can be obtained through numerical search within $0\leq\alpha\leq1$. Denote $f_{2}(\bm{\ell}^*)$ as the optimal value of Problem (P1') with given $\alpha$, we have
\begin{align}\label{eqn:r41}
&\alpha^*=\arg\min_{\alpha} f_{2}(\bm{\ell}^*),\ \mathrm{s.t.} \quad0\leq\alpha\leq1.
\end{align}
 Note that search in \eqref{eqn:r41} is using the exhaustive   search for obtaining optimal $\alpha$. That is, $\alpha$ is divided into sufficiently small intervals within $[0,1]$ and we peak one that minimizes $f_{2}(\bm{\ell}^*)$ in \eqref{eqn:r41}.

The whole algorithm solving Problem (P1) optimally is summarized in Algorithm \ref{alg:A1}. The complexity of Algorithm \ref{alg:A1} is evaluated as follows. The complexity of the ellipsoid method is $\mathcal{O}(N^{2})$, where $N$ is the number of dual variables and $N=3$ in this paper. The complexity for obtaining $\alpha$ is $K$ where $K$ is the resolution of the numerical search. Thus the total complexity of Algorithm \ref{alg:A1} is $\mathcal{O}(9K)$.

\begin{algorithm}[tb]
\caption{Optimal algorithm for Problem (P1)  }\label{alg:A1}

\begin{algorithmic}[1]
\STATE Initialize $\alpha$ and $\bm{\lambda}$.
\REPEAT
\STATE Compute $\ell_{\mathrm{u,a}}$, $\ell_{\mathrm{u,h}}$, and  $\ell_{\mathrm{h,a}}$ using \eqref{eqn:r35}, \eqref{eqn:r36}, and  \eqref{eqn:r38}, respectively.
\STATE Update $\bm{\lambda}$ by the ellipsoid method using subgradients \eqref{eqn:r39}, \eqref{eqn:r40}, and \eqref{eqn:subgradiant_lam3}.
\UNTIL{$\bm{\lambda}$ converges to a prescribed accuracy.}
\STATE Obtain $\alpha^{*}$ by \eqref{eqn:r41}.
\STATE Obtain $t_{\mathrm{u}}^{*}=\alpha^{*}T$ and $t_{\mathrm{h}}^{*}=(1-\alpha^{*})T$.
\end{algorithmic}
\end{algorithm}

\subsection{Special Case}
We consider a special case of $L_{\mathrm{h}}=0$. In this case, the helper has no task to execute and just helps computing the bits offloaded by the user, i.e., $\ell_{\mathrm{h,a}}=0$ and $t_{\mathrm{h}}=0$. The	energy consumption minimization problem can be simplified as
\begin{eqnarray}
\min_{\ell_{\mathrm{u,a}},\ell_{\mathrm{u,h}},t_{\mathrm{u}}}&&f_{4}(\ell_{\mathrm{u,a}},\ell_{\mathrm{u,h}},t_{\mathrm{u}})\nonumber \\
{\rm s.t.}&&\frac{(L_{\mathrm{u}}-\ell_{\mathrm{u,h}}-\ell_{\mathrm{u,a}})C_{\mathrm{u}}}{f_{\mathrm{u}}}\leq T,\nonumber\\
  && t_{\mathrm{u}}+\frac{\ell_{\mathrm{u,h}}C_{\mathrm{u}}}{f_{\mathrm{h}}} \leq T,\nonumber\\
  &&\ell_{\mathrm{u,a}}C_{\mathrm{u}} \leq F, \nonumber\\
  && \ell_{\mathrm{u,a}},\ell_{\mathrm{u,h}} \in \mathbb{R}_{\geq 0},\nonumber\\
   && 0\leq t_{\mathrm{u}}\leq T,\label{eqn:r42}
  \end{eqnarray}
where $f_{4}$ is defined to be the mapping
\begin{align}\label{eqn:r43}
f_{4}:\qquad\qquad\mathbb{R}_{\geq 0}^{3} \rightarrow\ & \mathbb{R}, \nonumber\\ 
(\ell_{\mathrm{u,a}},\ell_{\mathrm{u,h}},t_{\mathrm{u}})\mapsto\ &
w_{\mathrm{u}}\kappa f_{\mathrm{u}}^{2}(L_{\mathrm{u}}-\ell_{\mathrm{u,h}}-\ell_{\mathrm{u,a}})
+w_{\mathrm{h}}\kappa f_{\mathrm{h}}^{2}\ell_{\mathrm{u,h}}\nonumber\\&
+\left(\frac{w_{\mathrm{u}}t_{\mathrm{u}}}{h_{\mathrm{u,a}}}-\frac{w_{\mathrm{u}}t_{\mathrm{u}}}{h_{\mathrm{u,h}}}\right)f_{1}\left(\frac{\ell_{\mathrm{u,a}}}{t_{\mathrm{u}}}\right)
\nonumber\\&
+\frac{w_{\mathrm{u}}t_{\mathrm{u}}}
{h_{\mathrm{u,h}}}f_{1}
\left(\frac{\ell_{\mathrm{u,h}}
+\ell_{\mathrm{u,a}}}{t_{\mathrm{u}}}\right).
\end{align}
\begin{lemma}\label{Lm2}
The optimal solutions of problem \eqref{eqn:r42} satisfy $\ell_{\mathrm{u,h}}^{*}=(T-t_{\mathrm{u}}^*)f_{\mathrm{h}}/C_{\mathrm{u}}$.
\end{lemma}

\begin{proof}
 	Please refer to Appendix \ref{AP3}.
\end{proof}

Lemma \ref{Lm2} means that when the helper does not have its own task to execute, i.e., $L_{\mathrm{h}}=0$, the user tries to make its offloading time $t_{\mathrm{u}}$ as long as possible until the helper fails to complete the task offloaded by the user during the remaining time $T-t_{\mathrm{u}}$.

With $\ell_{\mathrm{u,h}}=(T-t_{\mathrm{u}})f_{\mathrm{h}}/C_{\mathrm{u}}$, we have
\begin{align}\label{eqn:r45}
\frac{\partial f_4}{\partial t_{\mathrm{u}}}=&
\left(\frac{w_{\mathrm{u}}}{h_{\mathrm{u,a}}}
-\frac{w_{\mathrm{u}}}{h_{\mathrm{u,h}}}\right)
\left(2^{\frac{\ell_{\mathrm{u,a}}}{t_{\mathrm{u}}}}
\left(1-\frac{\ell_{\mathrm{u,a}}\ln2}{t_{\mathrm{u}}}\right) 
-1\right)\nonumber\\&
+\frac{w_{\mathrm{u}}}{h_{\mathrm{u,h}}}2^{\frac{(T-t_{\mathrm{u}})f_{\mathrm{h}}+\ell_{\mathrm{u,a}}
C_{\mathrm{u}}}{t_{\mathrm{u}}C_{\mathrm{u}}}}\left(1-\frac{(Tf_{\mathrm{h}}
+\ell_{\mathrm{u,a}}C_{\mathrm{u}})\ln2}{t_{\mathrm{u}}C_{\mathrm{u}}}\right) 
\nonumber\\&
-\frac{w_{\mathrm{u}}}{h_{\mathrm{u,h}}}
+\frac{\kappa f_{\mathrm{h}}}{C_{\mathrm{u}}}\left(w_{\mathrm{u}}f_{\mathrm{u}}^{2}
-w_{\mathrm{h}}f_{\mathrm{h}}^{2}\right).
\end{align}
From \eqref{eqn:r45} we can see that there are two cases of $\frac{\partial f_4}{\partial t_{\mathrm{u}}}$. The first case is $w_{\mathrm{u}}f_{\mathrm{u}}^{2}-w_{\mathrm{h}}f_{\mathrm{h}}^{2}\leq 0$, which means that compared with computing task at the helper, it is more energy-efficient to locally compute the task at the user. As we proved in Appendix \ref{AP1}, in this case $\frac{\partial f_4}{\partial t_{\mathrm{u}}}\leq 0$, which denotes that the objective function $f_4$ is a non-increasing function with respect to $t_{\mathrm{u}}$. Thus we can further derive that $t_{\mathrm{u}}^*=T$ and $\ell_{\mathrm{u,h}}^*=0$, which means that the user only offloads data to the AP. Then with $t_{\mathrm{u}}=T$, the partial derivative of $f_4$ with respect to $\ell_{\mathrm{u,a}}$ can be expressed as
\begin{align}\label{eqn:r46}
\frac{\partial f_4}{\partial \ell_{\mathrm{u,a}}}=&\frac{w_{\mathrm{u}}\ln2}{h_{\mathrm{u,a}}}2^{\frac{\ell_{\mathrm{u,a}}}{T}}-w_{\mathrm{u}}\kappa f_{\mathrm{u}}^{2}.
\end{align}
Then $\ell_{\mathrm{u,a}}^*$ can be obtained by setting $\frac{\partial f_4}{\partial \ell_{\mathrm{u,a}}}$ equal to zero and considering the constraints $(L_{\mathrm{u}}-\ell_{\mathrm{u,a}})C_{\mathrm{u}}\leq f_{\mathrm{u}}T$ and $\ell_{\mathrm{u,a}}C_{\mathrm{u}} \leq F$. We have
\begin{align}\label{eqn:r47}
&\ell_{\mathrm{u,a}}^*=\min\left(\max \left(T\left[\log_{2}\frac{\kappa f_{\mathrm{u}}^{2}h_{\mathrm{u,a}}}{\ln2}\right]^{+}\!\!, L_{\mathrm{u}}-\frac{f_{\mathrm{u}}T}{C_{\mathrm{u}}}\right), \frac{F}{C_{\mathrm{u}}}\right).
\end{align}

On the other hand, for the case  $w_{\mathrm{u}}f_{\mathrm{u}}^{2}-w_{\mathrm{h}}f_{\mathrm{h}}^{2}> 0$, we use the Lagrangian duality method to solve the optimal values of $\ell_{\mathrm{u,a}}$ and $t_{\mathrm{u}}$ because of the convexity of the problem.  The Lagrangian function is given by
\begin{align}
& L_2(\ell_{\mathrm{u,a}}, t_{\mathrm{u}},\lambda_{1})=
\frac{w_{\mathrm{u}}t_{\mathrm{u}}}{h_{\mathrm{u,h}}}f_{1}\left(\frac{(T-t_{\mathrm{u}})f_{\mathrm{h}}
+\ell_{\mathrm{u,a}}C_{\mathrm{u}}}{t_{\mathrm{u}}C_{\mathrm{u}}}\right)
\nonumber\\&
+\left(\frac{w_{\mathrm{u}}t_{\mathrm{u}}}{h_{\mathrm{u,a}}}-\frac{w_{\mathrm{u}}t_{\mathrm{u}}}{h_{\mathrm{u,h}}}\right)f_{1}
\left(\frac{\ell_{\mathrm{u,a}}}{t_{\mathrm{u}}}\right)
+w_{\mathrm{h}}\kappa f_{\mathrm{h}}^{2}\frac{(T-t_{\mathrm{u}})
f_{\mathrm{h}}}{C_{\mathrm{u}}}\nonumber\\&+w_{\mathrm{u}}\kappa f_{\mathrm{u}}^{2}
(L_{\mathrm{u}}-(T-t_{\mathrm{u}})\frac{f_{\mathrm{h}}}{C_{\mathrm{u}}}
-\ell_{\mathrm{u,a}})\nonumber\\&+\lambda_{1}(L_{\mathrm{u}}C_{\mathrm{u}}
-(T-t_{\mathrm{u}})f_{\mathrm{h}}-\ell_{\mathrm{u,a}}C_{\mathrm{u}}-f_{\mathrm{u}}T),\label{eqn:r44}
\end{align}
where $\lambda_{1}\in \mathbb{R}_{\geq 0}$ is the Lagrangian multiplier with a slight abuse of notation. The dual function is defined to be $g_2(\lambda_{1}):=\min_{ F/C_{\mathrm{u}}\geq \ell_{\mathrm{u,a}}\geq 0,T\geq t_{\mathrm{u}}\geq0 }L_2(\ell_{\mathrm{u,a}}, t_{\mathrm{u}},\lambda_{1})$ and the dual problem is $\max_{ \lambda_{1}\geq 0}g_2(\lambda_{1})$.
Note that problem \eqref{eqn:r42} is a convex problem and the block coordinate descent (BCD) method \cite{Richtarik2014} can be adopted to solve the problem, where we alternatively optimize one of $\ell_{\mathrm{u,a}}$ and $t_{\mathrm{u}}$ with the other fixed.
Given $t_{\mathrm{u}}$, we have
\begin{align}
\frac{\partial L_2}{\partial \ell_{\mathrm{u,a}}}=&\frac{w_{\mathrm{u}}\ln2}{h_{\mathrm{u,h}}}2^{\frac{(T-t_{\mathrm{u}})f_{\mathrm{h}}+\ell_{\mathrm{u,a}}C_{\mathrm{u}}}{t_{\mathrm{u}}C_{\mathrm{u}}}}-w_{\mathrm{u}}\kappa f_{\mathrm{u}}^{2}-\lambda_{1}C_{\mathrm{u}}\nonumber\\&+w_{\mathrm{u}}\ln2\left(\frac{1}{h_{\mathrm{u,a}}} - \frac{1}{h_{\mathrm{u,h}}}\right)2^{\frac{\ell_{\mathrm{u,a}}}{t_{\mathrm{u}}}}.\label{eqn:r48}
\end{align}
We can obtain $\ell_{\mathrm{u,a}}^*$ by setting $\frac{\partial L_2}{\partial \ell_{\mathrm{u,a}}}$ equal to zero and considering $F/C_{\mathrm{u}}\geq \ell_{\mathrm{u,a}}$, which is expressed as
\begin{align}\label{eqn:r65}
	&\ell_{\mathrm{u,a}}^*=
	\min\biggl( \frac{F}{C_{\mathrm{u}}},\nonumber\\&
	 t_{\mathrm{u}}\left(\log_{2}
	\left(\frac{h_{\mathrm{u,h}}h_{\mathrm{u,a}}
	(w_{\mathrm{u}}\kappa f_{\mathrm{u}}^{2}
	+\lambda_{1}C_{\mathrm{u}})}{w_{\mathrm{u}}
	\ln2(h_{\mathrm{u,a}}2^{\frac{f_{\mathrm{h}}(T-t_{\mathrm{u}})}
	{t_{\mathrm{u}}C_{\mathrm{u}}}}+h_{\mathrm{u,h}}
	-h_{\mathrm{u,a}})}\right)\right)^{+}
	 &\!\!\biggr).
\end{align}

Next, with fixed $\ell_{\mathrm{u,a}}$, we have $\frac{\partial L_2}{\partial t_{\mathrm{u}}}=\frac{\partial f_4}{\partial t_{\mathrm{u}}}+\lambda_{1}f_{\mathrm{h}}$, where $\frac{\partial f_4}{\partial t_{\mathrm{u}}}$ can be obtained in \eqref{eqn:r45}. According to KKT conditions, $t_{\mathrm{u}}^*$ can be obtained by setting $\frac{\partial L_2}{\partial t_{\mathrm{u}}}$ equal to zero. However, the closed-form solution of $t_{\mathrm{u}}$ is non-trivial to obtain and thus we adopt the bisection search within $0\leq t_{\mathrm{u}}\leq T$ to solve $t_{\mathrm{u}}^*$. After solving the dual function $g_2(\lambda_{1})$, we adopt the bisection search to find the optimal $\lambda_{1}^*$.

The algorithm for solving problem \eqref{eqn:r42} is presented in Algorithm \ref{alg:A3}. Note that the complexities of the bisection method used for obtaining $ t_{\mathrm{u}}^*$ and $\lambda_{1}^*$ are sub-linear. Moreover, the complexity for solving $\ell_{\mathrm{u,a}}^*$ and $t_{\mathrm{u}}^*$ with the BCD method is linear.  Thus the total complexity of Algorithm \ref{alg:A3} is linear.

\begin{algorithm}[t]
\caption{Optimal algorithm for problem \eqref{eqn:r42}}\label{alg:A3}

\begin{algorithmic}[1]
\IF {$w_{\mathrm{u}}f_{\mathrm{u}}^{2}-w_{\mathrm{h}}f_{\mathrm{h}}^{2}\leq 0$}
\STATE Compute $\ell_{\mathrm{h,a}}^*$ by \eqref{eqn:r47}.
\STATE Set $t_{\mathrm{u}}^*=T$.
\ELSE
\STATE Initialize $\lambda_{1}$.
\REPEAT
\STATE Initialize $\ell_{\mathrm{u,a}}$ and $t_{\mathrm{u}}$.
\REPEAT
\STATE Compute $\ell_{\mathrm{u,a}}$ by \eqref{eqn:r65} for given $t_{\mathrm{u}}$.
\STATE Compute $ t_{\mathrm{u}}$ that maximizes $L_2$ by bisection search for given $\ell_{\mathrm{u,a}}$.
\UNTIL{The improvement of $L_{2}$ stops.}
\STATE Update $\lambda_{1}$ by the bisection method.
\UNTIL{$\lambda_{1}$ converges to a prescribed accuracy.}
\ENDIF
\STATE Obtain $\ell_{\mathrm{u,h}}^{*}=(T-t^{*}_{\mathrm{u}})f_{\mathrm{h}}/C_{\mathrm{u}}$.
\end{algorithmic}
\end{algorithm}

\section{Optimal Solution for Offloading Data Maximization Problem}\label{se3}
In this section, we solve the offloading data maximization Problem (P2).
\subsection{Problem Transformation}
\begin{lemma}\label{Lm4}
	The optimal transmit power of the user satisfies $p_{\mathrm{u,h}}^*+p_{\mathrm{u,a}}^*=\bar{P}_{\mathrm{u}}$.
\end{lemma}
\begin{proof}
 	Please refer to Appendix \ref{AP4}.
\end{proof}
For simplicity, we introduce a variable $\beta\in \mathbb{R}_{\geq 0}$ denoting the proportion of $p_{\mathrm{u,h}}$ over $\bar{P}_{\mathrm{u}}$. Then $p_{\mathrm{u,h}}$ and $p_{\mathrm{u,a}}$ can be rewritten as $p_{\mathrm{u,h}}=\beta\bar{P}_{\mathrm{u}}$ and $p_{\mathrm{u,a}}=(1-\beta)\bar{P}_{\mathrm{u}}$, respectively. Problem (P2) is equivalent to
\begin{align}
{\text{ (P2')}:}\max_{\beta, \bm{t}}&\quad  w_{\mathrm{u}}t_{\mathrm{u}}
\left(\log_{2}\frac{1+\beta\bar{P}_{\mathrm{u}}h_{\mathrm{u,h}}}
{1+\beta\bar{P}_{\mathrm{u}}h_{\mathrm{u,a}}}+R_{1}\right)\!
+\! w_{\mathrm{h}}t_{\mathrm{h}}R_{2}\nonumber \\
{\rm s.t.}
  &\quad  \kappa f^{2}_{\mathrm{h}} t_{\mathrm{u}}\log_{2} \left(1+\beta\bar{P}_{\mathrm{u}} 
  h_{\mathrm{u,h}}\right) \leq E'_{\mathrm{h}},\label{eqn:const_enegy_helper} \\
  &\quad  t_{\mathrm{u}}\left(1+\frac{C_{\mathrm{u}}}{f_{\mathrm{h}}} 
  \log_{2}\left(1+\beta\bar{P}_{\mathrm{u}}h_{\mathrm{u,h}}\right)\right)\leq T,\label{eqn:r69}\\
  &\quad  0\leq t_{\mathrm{u}}+t_{\mathrm{h}}\leq T,\label{eqn:r68}\\
  &\quad  0\leq \beta \leq 1, \label{eqn:r49}\\
  &\quad  \beta\in \mathbb{R}, \bm{t}\in \mathbb{R}^{2},
\end{align}
where $R_{1}:=\log_{2}\left(1+h_{\mathrm{u,a}}\bar{P}_{\mathrm{u}}\right)$ and $R_{2}:=\log_{2}\left(1+h_{\mathrm{h,a}}\bar{P}_{\mathrm{h}}\right)$ respectively denote the maximum offloading rates from the user and helper to the AP. Note that the objective function of Problem (P2') is non-decreasing with respect to $t_{\mathrm{h}}$. Therefore, to maximize the weighted offloading data, the equality $t_{\mathrm{u}}+t_{\mathrm{h}}=T$ holds.

Additionally, the objective function of Problem (P2') is non-decreasing with respect to $\beta$ due to the assumption $h_{\mathrm{u,h}}\geq h_{\mathrm{u,a}}$. Since the upper bound of $\beta$ is related to constraints \eqref{eqn:const_enegy_helper}, \eqref{eqn:r69}, and \eqref{eqn:r49}, at least one of the equalities (or upper bounds) in the constraints \eqref{eqn:const_enegy_helper}, \eqref{eqn:r69}, and \eqref{eqn:r49} holds at the optimal point of $\beta$. This property can be obtained by contradiction and the details are omitted here.

First we consider the case $\beta = 1$, i.e., the equality in constraint \eqref{eqn:r49} holds, which implies that the user allocates all transmit power to the helper. Then Problem (P2') turns out to be a time allocation problem where $t_{\mathrm{u}}$ and $t_{\mathrm{h}}$ are jointly optimized. With $t_{\mathrm{u}}+t_{\mathrm{h}}=T$, Problem (P2') can be simplified as
\begin{align}
\max_{ t_{\mathrm{u}} }&\quad w_{\mathrm{u}}t_{\mathrm{u}}\log_{2}\left(1+\bar{P}_{\mathrm{u}}h_{\mathrm{u,h}}\right)+w_{\mathrm{h}}R_{2}(T-t_{\mathrm{u}})\label{Rate_beta1}\\
{\rm s.t.}
  &\quad \kappa f^{2}_{\mathrm{h}} t_{\mathrm{u}}\log (1+\bar{P}_{\mathrm{u}} h_{\mathrm{u,h}}) \leq E'_{\mathrm{h}},\\
  &\quad t_{\mathrm{u}}\left(1+\frac{C_{\mathrm{u}}}{f_{\mathrm{h}}} 
  \log_{2}\left(1+\bar{P}_{\mathrm{u}}h_{\mathrm{u,h}}\right)\right)\leq T,\label{eqn:r54}\\
  &\quad  t_{\mathrm{u}}\in \mathbb{R}_{\geq0}.
  \end{align}

Then we obtain the solution of the above problem:
\begin{align}\label{eqn:r55}
t_{\mathrm{u}}^*|_{\beta=1}=\left\{\begin{array}{ll}
0&  E\leq 0,\\
t_{1}&
\text{otherwise},
\end{array}\right.
\end{align}
where $E:=w_{\mathrm{u}}\log_{2}\left(1+\bar{P}_{\mathrm{u}}h_{\mathrm{u,h}}\right)-w_{\mathrm{h}}R_{2} \in \mathbb{R}$ is a constant with respect to $t_{\mathrm{u}}$ and $t_{1}$ is defined to be
\begin{align*}
t_{1}:=&\min\left(\frac{Tf_{\mathrm{h}}}{f_{\mathrm{h}}+C_{\mathrm{u}}
\log_{2}\left(1+\bar{P}_{\mathrm{u}}h_{\mathrm{u,h}}\right)},\right.
\nonumber\\&\left.\frac{E'_{\mathrm{h}}}{\kappa f^{2}_{\mathrm{h}}
\log (1+\bar{P}_{\mathrm{u}} h_{\mathrm{u,h}})}\right).
\end{align*}
Define $R_{\mathrm{total}}(\beta=1)$ to be the optimal value of problem \eqref{Rate_beta1}, we have 
\begin{align}\label{eqn:r56}
R_{\mathrm{total}}(\beta=1)=
w_{\mathrm{h}}TR_{2}+[E]^{+}t_{1}.
\end{align}

Let us consider another case that the equality in the constraint \eqref{eqn:r69} holds, implying that the user offloads as much as computation input data to the helper during the whole transmit time $T$, i.e., the time of receiving and processing the user's task in the helper is $T$. In this case, we can replace $\beta$ by function $\beta_{1}(t_{\mathrm{u}})$, which is defined to be
\begin{align}
  \beta_{1}(t_{\mathrm{u}}):= \frac{2^{\frac{f_{\mathrm{h}}(T-t_{\mathrm{u}})}{C_{\mathrm{u}} t_{\mathrm{u}}}}-1}{\bar{P}_{\mathrm{u}}h_{\mathrm{u,h}}},\ \forall t_{\mathrm{u}}\in \mathbb{R}_{\geq 0}.\label{eqn:r73}
\end{align}
Because of the constraint $0\leq \beta\leq 1$, the condition $0\leq\beta_{1}(t_{\mathrm{u}})\leq 1$ is satisfied if and only if $t_{\mathrm{u}}\geq Tf_{\mathrm{h}}/(f_{\mathrm{h}}+C_{\mathrm{u}}\log_{2}\left(1+\bar{P}_{\mathrm{u}}h_{\mathrm{u,h}}\right)$.

From \eqref{eqn:r73}, we observe that $\beta_{1}(t_{\mathrm{u}})$ is an exponential function with respect to $1/t_{\mathrm{u}}$. In other words, when the user offloads more bits to the helper (i.e., $\beta_{1}(t_{\mathrm{u}})$ increases), it spends less time in offloading data (i.e., smaller $t_{\mathrm{u}}$) and thus the helper has a longer time ($t_{\mathrm{h}}=T-t_{\mathrm{u}}$) to offload data. This explicitly explains that in this condition the proposed cooperative computing scheme leads to a win-win situation.

After substituting \eqref{eqn:r73} into Problem (P2'), Problem (P2') is equivalent to
\begin{align}
\max_{t_{\mathrm{u}}}&\quad f_{5}(t_{\mathrm{u}})\label{eqn:r50} \\
{\rm s.t.}
  &\quad \frac{Tf_{\mathrm{h}}}
  {f_{\mathrm{h}}+C_{\mathrm{u}}\log_{2}
  \left(1+\bar{P}_{\mathrm{u}}h_{\mathrm{u,h}}\right)} 
  \leq t_{\mathrm{u}}\\
 &\quad T-\frac{E'_{\mathrm{h}}C_{\mathrm{u}}}
 	{\kappa f^{3}_{\mathrm{h}}} \leq t_{\mathrm{u}}\\
  &\quad t_{\mathrm{u}}\leq T\\
  &\quad t_{\mathrm{u}}\in \mathbb{R},
  \end{align}
where $f_{5}$ is defined to be the mapping
\begin{align}
f_{5}: \mathbb{R}_{\geq 0}\rightarrow&\ \mathbb{R}, \nonumber\\ 
t_{\mathrm{u}}\mapsto&\ t_{\mathrm{u}}\left(w_{\mathrm{u}}R_{1}
-\frac{w_{\mathrm{u}}f_{\mathrm{h}}}{C_{\mathrm{u}}}-w_{\mathrm{h}}R_{2}\right)\nonumber\\&
+T\left(\frac{w_{\mathrm{u}}f_{\mathrm{h}}}{C_{\mathrm{u}}}+w_{\mathrm{h}}R_{2}\right)\nonumber\\&
-w_{\mathrm{u}}t_{\mathrm{u}}\log_{2}\left(1-\frac{h_{\mathrm{u,a}}}{h_{\mathrm{u,h}}}
+\frac{h_{\mathrm{u,a}}}{h_{\mathrm{u,h}}}2^{\frac{f_{\mathrm{h}}(T-t_{\mathrm{u}})}{C_{\mathrm{u}}
 t_{\mathrm{u}}}}\right).\label{eqn:r64}
\end{align}

The derivative of $f_{5}(t_{\mathrm{u}})$ is
\begin{align}
\dv{f_{5}}{t_{\mathrm{u}}}=&-w_{\mathrm{u}}\log_{2}\left(1-\frac{h_{\mathrm{u,a}}}
{h_{\mathrm{u,h}}}
+\frac{h_{\mathrm{u,a}}}
{h_{\mathrm{u,h}}}2^{\frac{f_{\mathrm{h}}
(T-t_{\mathrm{u}})}{C_{\mathrm{u}} t_{\mathrm{u}}}}\right)
-w_{\mathrm{u}}\frac{f_{\mathrm{h}}}{C_{\mathrm{u}}}\nonumber\\&
+\frac{w_{\mathrm{u}}Tf_{\mathrm{h}}h_{\mathrm{u,a}}2^{\frac{f_{\mathrm{h}}
(T-t_{\mathrm{u}})}{C_{\mathrm{u}} t_{\mathrm{u}}}}}{C_{\mathrm{u}} 
t_{\mathrm{u}}(h_{\mathrm{u,h}}-h_{\mathrm{u,a}}+h_{\mathrm{u,a}}
2^{\frac{f_{\mathrm{h}}(T-t_{\mathrm{u}})}{C_{\mathrm{u}} t_{\mathrm{u}}}})}\nonumber\\&
-w_{\mathrm{h}}R_{2}+w_{\mathrm{u}}R_{1}.\label{eqn:r53}
\end{align}
\begin{lemma}\label{lem5}
	Function $f_{5}(t_{\mathrm{u}})$ is a concave function over the feasible set of problem \eqref{eqn:r50}.
\end{lemma}
As we prove in Appendix E, $(\dv{}{t_{\mathrm{u}}})^{2}f_{5}\leq 0$ is always satisfied. Let us define $R_{\mathrm{total}}(\beta=\beta_{1}(t_{\mathrm{u}}))$ to be the optimal value of problem \eqref{eqn:r50}. We can solve $R_{\mathrm{total}}(\beta=\beta_{1}(t_{\mathrm{u}}))$ by solving $\dv{f_{5}}{t_{\mathrm{u}}}=0$ and considering constraints of $t_{\mathrm{u}}$ in problem \eqref{eqn:r50}. However, it is non-trivial to obtain the closed-form solution of $t_{\mathrm{u}}$. Here we use the bisection search to obtain the optimal $t_{\mathrm{u}}$.

As for the case that constraint \eqref{eqn:const_enegy_helper} holds, i.e., the helper's available energy for computing the task offloaded by the user is fully utilized, the variable $\beta$ can be represented by function $\beta_{2}(t_{\mathrm{u}})$, which is defined to be
\begin{align}
  \beta_{2}(t_{\mathrm{u}}):=\frac{2^{\frac{E'_{\mathrm{h}}}{\kappa f_{\mathrm{h}}^{2}t_{\mathrm{u}} }}-1}{\bar{P}_{\mathrm{u}}h_{\mathrm{u,h}}},\ \forall t_{\mathrm{u}}\in \mathbb{R}_{\geq 0}.
\end{align}
The condition $0\leq\beta_{2}(t_{\mathrm{u}})\leq 1$ is satisfied if and only if the inequality $E'_{\mathrm{h}}/(\kappa f^{2}_{\mathrm{h}}\log_{2} (1+\bar{P}_{\mathrm{u}}h_{\mathrm{u,h}}))\leq t_{\mathrm{u}}$ holds for all $t_{\mathrm{u}}\geq 0$. Similarly, from the definition of $\beta_{2}(t_{\mathrm{u}})$, we find that this case is also a win-win situation between the user and the helper. 

With $\beta=\beta_{2}(t_{\mathrm{u}})$, Problem (P2') can be rewritten as
\begin{align}
\max_{t_{\mathrm{u}}}&\quad f_{7}(t_{\mathrm{u}})\label{eqn:tmin_beta_third_case} \\
{\rm s.t.}
  &\quad \frac{E'_{\mathrm{h}} }{\kappa f^{2}_{\mathrm{h}}\log (1+\bar{P}_{\mathrm{u}}h_{\mathrm{u,h}})}\leq t_{\mathrm{u}}
  \leq T-\frac{E'_{\mathrm{h}}C_{\mathrm{u}}}{\kappa f^{3}_{\mathrm{h}}},\\
  &\quad t_{\mathrm{u}}\leq T,
\end{align}
where $f_{7}$ is defined to be the mapping 
\begin{align}
f_{7}: \mathbb{R}_{\geq 0}\rightarrow\ & \mathbb{R}, \nonumber\\
 t_{\mathrm{u}}\mapsto\ &\frac{w_{\mathrm{u}} E'_{\mathrm{h}}}
 {\kappa f^{2}_{\mathrm{h}}}+w_{\mathrm{h}}TR_{2}+t_{\mathrm{u}}
 \left(w_{\mathrm{u}}R_{1}-w_{\mathrm{h}}R_{2}\right)\nonumber\\&
 -w_{\mathrm{u}}t_{\mathrm{u}}\log_{2}\left(1-\frac{h_{\mathrm{h,a}}}
 {h_{\mathrm{u,h}}}+\frac{h_{\mathrm{h,a}}}{h_{\mathrm{u,h}}}
 2^{\frac{E'_{\mathrm{h}}}{\kappa f^{2}_{\mathrm{h}}t_{\mathrm{u}} }}\right).\label{eqn:obj_third_case}
\end{align}
Similarly as Appendix E, we have $(\dv{}{t_{\mathrm{u}}})^{2}f_{7}\leq 0$ for all $t_{\mathrm{u}}\geq 0$. We can obtain the optimal $t_{\mathrm{u}}$ by solving $\dv{f_{7}}{t_{\mathrm{u}}}=0$ and considering constraints of $t_{\mathrm{u}}$ in problem \eqref{eqn:tmin_beta_third_case}. Here we use the bisection search to obtain the optimal $t_{\mathrm{u}}$.
 
By defining $R_{\mathrm{total}}(\beta=\beta_{2}(t_{\mathrm{u}}))$ to be the optimal value of problem \eqref{eqn:tmin_beta_third_case}, the optimal value of Problem (P2) equals $R_{\mathrm{total}}^{*}\in \mathbb{R}$, which is defined to be 
\begin{align}
  R_{\mathrm{total}}^{*}:=&
  \max\left(R_{\mathrm{total}}(\beta=1),
  R_{\mathrm{total}}(\beta= \beta_{1}(t_{\mathrm{u}})),\right.\nonumber\\&\left.
  R_{\mathrm{total}}(\beta= \beta_{2}(t_{\mathrm{u}}))\right).
\end{align}
The whole algorithm for addressing Problem (P2) optimally is summarized in Algorithm \ref{alg:A2}.

In summary, the optimal solution of Problem (P2) has a special structure, in which only three cases occur. The first case is that the user only offloads data to the helper due to the short distance and high rate. The second case is that the helper fully uses the whole transmit time for receiving and then processing user's offloaded data. The third case is that all helper's available energy for processing user's offloaded data is used.
\begin{algorithm}[tb]
\caption{Optimal algorithm for Problem (P2)  }\label{alg:A2}

\begin{algorithmic}[1]
\STATE Obtain $R_{\mathrm{total}}^{*}$ by solving \eqref{eqn:r56}, \eqref{eqn:r50}, and \eqref{eqn:tmin_beta_third_case}.
\IF {$R_{\mathrm{total}}^{*}=R_{\mathrm{total}}(\beta=1)$}
\STATE Obtain $t_{\mathrm{u}}^{*}$ by \eqref{eqn:r55}.
\STATE $\beta^* = 1$.
\ELSIF {$R_{\mathrm{total}}^{*}=R_{\mathrm{total}}(\beta=\beta_{1}(t_{\mathrm{u}}) )$}
\STATE Obtain $t_{\mathrm{u}}^{*}$ by solving problem \eqref{eqn:r50}.
\STATE $\beta^* = \beta_{1}(t^{*}_{\mathrm{u}})$.
\ELSE 
\STATE Obtain $t_{\mathrm{u}}^{*}$ by solving problem \eqref{eqn:tmin_beta_third_case}.
\STATE $\beta^* = \beta_{2}(t^{*}_{\mathrm{u}})$.
\ENDIF
\STATE Set $t_{\mathrm{h}}^*=T-t_{\mathrm{u}}^*$.
\end{algorithmic}
\end{algorithm}

\subsection{Special Case}
\subsubsection{High Signal-to-noise}
Here we consider a special case of high signal-to-noise (SNR) for Problem (P2). Note that Lemma \ref{Lm4} is applicable for this case as well. Then Problem (P2) can be simplified as
\begin{align}\label{eqn:r70}
\max_{\beta,\bm{t}\geq \bm{0}}&\quad w_{\mathrm{u}}t_{\mathrm{u}}\left(\log_{2}\frac{h_{\mathrm{u,h}}}{h_{\mathrm{u,a}}}+R_{1}\right)+w_{\mathrm{h}}t_{\mathrm{h}}R_{2} \\
{\rm s.t.}
 &\quad \kappa f^{2}_{\mathrm{h}} t_{\mathrm{u}}\log (1+\beta\bar{P}_{\mathrm{u}} h_{\mathrm{u,h}}) \leq E'_{\mathrm{h}},\label{eqn:const_enegy_helper2}\\
  &\quad t_{\mathrm{u}}\left(1+
  \frac{C_{\mathrm{u}}}{f_{\mathrm{h}}} \log_{2}
  \left(1+\beta\bar{P}_{\mathrm{u}}
  h_{\mathrm{u,h}}\right)\right)\leq T,\label{eqn:r169}\\
  &\quad t_{\mathrm{u}}+t_{\mathrm{h}}\leq T,\label{eqn:r168}\\
  &\quad 0\leq \beta \leq 1 \label{eqn:r149}.
  \end{align}
We can see that the power allocation proportion factor $\beta$ is eliminated in the objective function and only exists in the constraint \eqref{eqn:r169}. To maximize the objective value of the above problem, it is straightforward that the optimal value of $\beta$ is $\beta^*=0$ in the constraints \eqref{eqn:r169} and \eqref{eqn:const_enegy_helper2} so that $\bm{t}$ has a larger feasible set. Thus in this case, the user spends all of its own transmit power in offloading data to the AP. It is reasonable because the channel between the user and AP becomes strong under the high SNR condition and the computation latency at the AP is negligible, while the helper is restricted by the computation capacity (i.e., constraint  \eqref{eqn:r169}) and the limit of energy (constraint \eqref{eqn:const_enegy_helper2}) although its channel is also strong. As a result, the user only offloads data to the AP.

So with $\beta^*=0$,  problem \eqref{eqn:r70} turns out to be a linear programming problem and the optimal  solution can be illustrated as
\begin{align}
t_{\mathrm{u}}^*=\left\{\begin{array}{ll}
0&  	w_{\mathrm{u}}(\log_{2}\frac{h_{\mathrm{u,h}}}{h_{\mathrm{u,a}}}+R_{1})\leq w_{\mathrm{h}}R_{2},\\
T
&
\text{otherwise},\label{eqn:r71}
\end{array}\right.
\end{align}
and $t_{\mathrm{h}}^* = T-t_{\mathrm{u}}^*$.

The solution in \eqref{eqn:r71} shows that, if the user and helper's offloading channels are in high SNR, the optimal strategy is enabling one of the user and helper to occupy the whole time to offload data to the AP.

\subsubsection{Same Channel Condition From User to Helper and AP}\label{sameLoc}
Here we provide some theoretical analysis of our proposed cooperative scheme and two benchmark schemes so-called TDMA-based offloading scheme and NOMA-aided offloading scheme under the case that the helper and AP are at the same location. The details of TDMA-based offloading scheme and NOMA-aided offloading scheme are shown in Section \ref{se4}.

 When the distance between the helper and the AP is close to zero, we have $h_{\mathrm{h,a}}\rightarrow +\infty$ and $h_{\mathrm{h,a}}\gg h_{\mathrm{u,a}}=h_{\mathrm{u,h}}$.  In terms of the proposed scheme, it can be verified that Lemma \ref{Lm4} is also applicable for this case. With $h_{\mathrm{u,a}}=h_{\mathrm{u,h}}$ and Lemma \ref{Lm4}, Problem (P2) can be rewritten as
\begin{eqnarray}
\max_{\bm{t}\geq \bm{0}}&&w_{\mathrm{u}}t_{\mathrm{u}}R_{1}+w_{\mathrm{h}}t_{\mathrm{h}}R_{2}\nonumber \\
{\rm s.t.}
  &&t_{\mathrm{u}}\left(1+\frac{C_{\mathrm{u}}}{f_{\mathrm{h}}} \log_{2}\left(1+
  \beta\bar{P}_{\mathrm{u}}h_{\mathrm{u,h}}
  \right)\right)\leq T\label{eqn:r180}, \\
  && \kappa f^{2}_{\mathrm{h}} t_{\mathrm{u}}\log (1+\beta\bar{P}_{\mathrm{u}} h_{\mathrm{u,h}}) \leq E'_{\mathrm{h}},\label{eqn:const_enegy_helper3}\\
  &&t_{\mathrm{u}}+t_{\mathrm{h}}\leq T.
\end{eqnarray}
 To have a larger feasible set for variable $\bm{t}$ in the above problem, the optimal solution of $\beta$ is $\beta^*=0$ in constraints \eqref{eqn:r180} and \eqref{eqn:const_enegy_helper3}. Then this problem can be further simplified as
 \begin{eqnarray}
\max_{\bm{t}\geq \bm{0}}&&w_{\mathrm{u}}t_{\mathrm{u}}R_{1}+w_{\mathrm{h}}t_{\mathrm{h}}R_{2}\nonumber \\
{\rm s.t.}
  && t_{\mathrm{u}}+t_{\mathrm{h}}\leq T, \label{eqn:r150}
\end{eqnarray}
which is the same as the TDMA-based offloading data maximization problem.

As for the NOMA-aided offloading scheme without cooperation, we have two decoding orders. When $h_{\mathrm{u,a}}p_{\mathrm{u,a}}> h_{\mathrm{h,a}}p_{\mathrm{h,a}}$, we have $p_{\mathrm{h,a}}= 0$ because $h_{\mathrm{h,a}}\gg h_{\mathrm{u,a}}$ and $h_{\mathrm{h,a}}\rightarrow +\infty$. In this case the AP only receives signal offloaded by the user. In order to maximize the offloading data, the offloading rate for the user is $R_{1}$ in this condition. 

As for $h_{\mathrm{h,a}}p_{\mathrm{h,a}}\geq  h_{\mathrm{u,a}}p_{\mathrm{u,a}}$, the AP decodes the signal offloaded by helper first and then decodes the signal offloaded by the user using successive interference cancellation (SIC). The offloading rates of the user and the helper can be respectively written as 
\begin{align}
&R^{\mathrm{NOMA}}_{\mathrm{u,a}}=\log_{2}\left(1+h_{\mathrm{u,a}}p_{\mathrm{u,a}}\right),\\
&R^{\mathrm{NOMA}}_{\mathrm{h,a}}=\log_{2}\left(1+\frac{h_{\mathrm{h,a}}p_{\mathrm{h,a}}}{1+h_{\mathrm{u,a}}p_{\mathrm{u,a}}}\right).
\end{align}
The offloading data maximization problem for this decoding order can be formulated as
\begin{eqnarray}
\max_{p_{\mathrm{u,a}},p_{\mathrm{h,a}}}&& w_{\mathrm{u}}TR^{\mathrm{NOMA}}_{\mathrm{u,a}}+w_{\mathrm{h}}TR^{\mathrm{NOMA}}_{\mathrm{h,a}}\nonumber \\
{\rm s.t.}&& 0\leq p_{\mathrm{u,a}}\leq \bar{P}_{\mathrm{u}},\nonumber\\
&&0\leq p_{\mathrm{h,a}}\leq \bar{P}_{\mathrm{h}},\nonumber\\
&& h_{\mathrm{h,a}}p_{\mathrm{h,a}}\geq h_{\mathrm{u,a}}p_{\mathrm{u,a}}.\label{eqn:r152}
\end{eqnarray}
Obviously the optimal value of $p_{\mathrm{h,a}}$ is $\bar{P}_{\mathrm{h}}$ for problem \eqref{eqn:r152}. Since $h_{\mathrm{h,a}}\gg h_{\mathrm{u,a}}$ and $h_{\mathrm{h,a}}\rightarrow +\infty$, the optimal value of this problem is approximately as $w_{\mathrm{h}}TR_{2}$. 
Then the maximum offloading data for the NOMA-aided offloading scheme without cooperation is 
\begin{align}
\max(w_{\mathrm{u}}TR_{1}, w_{\mathrm{h}}TR_{2}), \label{eqn:r153}
\end{align}
which has the same value as the optimal value of problem \eqref{eqn:r150}.

Since $h_{\mathrm{h,a}}\rightarrow +\infty$ and $w_{\mathrm{u}}TR_{1}$ is a finite number, the optimal values of problem \eqref{eqn:r150} and problem \eqref{eqn:r153} equal $w_{\mathrm{h}}TR_{2}$. 
\section{Simulation Results}\label{se4}
\begin{table}[t]
\renewcommand\arraystretch{1.2}
\caption{\\Simulation Parameters} 
\centering 
\begin{tabular}{ | l | l|p{'2'}|} 
\hline
Bandwidth & 1 MHz \\ %
\hline
Distance between the user and AP & 150 meters\\
\hline
Distance between the helper and AP & 80 meters\\
\hline
Effective capacitance coefficient & $\kappa=10^{-26}$\\
\hline 
Helper's computation intensities  & $C_{\mathrm{u}}=1$ cycle/bit \\  %
\hline  
Helper's CPU frequency & 1 GHz\\ 
\hline  
\specialcell[c]{Helper's available energy for processing user's data} 
&\specialcell[c]{ $E'_{\mathrm{h}}=10^{-3}$Joule}\\
\hline 
Noise power & $\sigma^{2}=-$120 dBm \\  
\hline
Number of channel realizations & 2000\\
\hline 
\specialcell[t]{Path loss at a reference distance of 1 meter} & $10^{-3}$ \\
\hline 
Path loss exponent & 3 \\
\hline  
Power constraint for the user & $\bar{P}_{\mathrm{u}}=0.4$W\\ %
\hline
Power constraint for the helper & $\bar{P}_{\mathrm{h}}=0.8$W\\
\hline
Server's available computational capacity & $F=100$Kbits\\
\hline 
User's computation intensities & $C_{\mathrm{u}}=1$ cycle/bit \\  
\hline  
User's CPU frequency & 3 GHz\\ %
\hline
Weight of the user and helper & $w_{\mathrm{u}}=w_{\mathrm{h}}=1$ \\ %
\hline 
\end{tabular}
\label{table:sim_para} 
\end{table}

In this section, we provide simulation results to evaluate the proposed algorithms for the energy consumption minimization and offloading data maximization problems, respectively. The channel power gain is modeled as $g=cd^{-\phi}|\rho|^{2}$, where $c\in \mathbb{R}$ is the path loss at a reference distance of 1 meter and $d^{-\phi}$ denotes the propagation path loss with the path loss exponent $\phi\in \mathbb{R}_{\geq 0}$ and the distance $d\in \mathbb{R}_{\geq 0}$. Moreover, $\rho\sim \mathcal{CN}(0,1)$ denotes channel small-scale fading and $|\rho|$ follows Rayleigh distribution. The main system parameters are listed in Table \ref{table:sim_para}. For comparison purpose, we also consider the performance of the following benchmark schemes without cooperation:

\begin{enumerate}
	\item \textbf{TDMA-based offloading scheme (TDMA)}. In this scheme, the user and the helper respectively offload their tasks to the AP using TDMA. During the first slot $t_{\mathrm{u}}$, the user offloads a part of its input task to the AP. Then the helper is allowed to offload its input task in the remaining time. The proposed optimal algorithms (Algorithm \ref{alg:A1} and Algorithm \ref{alg:A2}) are also applicable for this case by  setting $p_{\mathrm{u,h}}$ and $\ell_{\mathrm{u,h}}$ equal to zeros.
	\item \textbf{NOMA-aided offloading scheme (NOMA)}. Here we also consider the NOMA-aided offloading without cooperation, where both the user and the helper partially offload their tasks to the AP based on NOMA. The solution for the energy consumption minimization problem can be found in \cite{Wang2017}. As for the sum data maximization problem, it is a simple convex optimization problem which can be easily solved  and the details are omitted here.
	\end{enumerate}
\subsection{Energy Consumption Minimization Case}
\begin{figure}[t]
\begin{centering}
\includegraphics[scale=0.5]{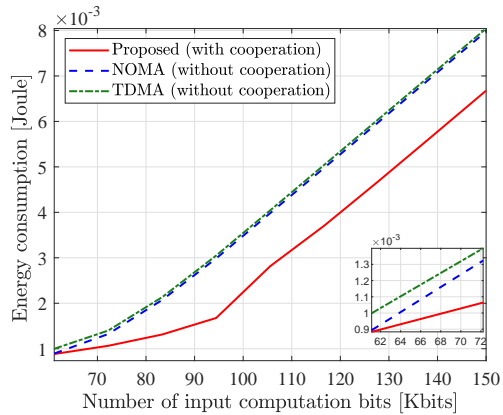}
\vspace{-0.1cm}
 \caption{ The average energy consumption versus $L_{\mathrm{u}}$. }\label{fig:E_min_Lu}
\end{centering}
\vspace{-0.1cm}
\end{figure}

In Fig. \ref{fig:E_min_Lu}, we plot the energy consumption  as a function of the input-bits of the user $L_{\mathrm{u}}$, with given $L_{\mathrm{h}}=80$ Kbits and $T=5$ ms. The distance between the user and the helper is fixed as 70 meters. First, we can observe that the proposed scheme is superior to the other benchmark schemes without cooperation, which validates the effectiveness of the proposed scheme. In particular, the proposed scheme has about $24\%$ energy consumption reduction on average, compared with the benchmark schemes. The energy consumption of  all schemes experiences a moderate increase with the growing number of input-bits $L_{\mathrm{u}}$. Note that the NOMA-aided offloading scheme consumes less energy than the TDMA-based offloading scheme. This is because in the NOMA-aided offloading scheme, both the user and the helper share the same resource block, which leads to a performance gain in the energy consumption. Moreover, the reasons for the performance gap among the two benchmark schemes and the proposed scheme can be explained as follows. In the proposed scheme, the user is allowed to offload data simultaneously to the AP and the helper, while in the benchmark schemes cooperative computing is not allowed. Note that the helper is closer to the user than the AP and thus it is more efficient for the user to offload a part of task to the helper, resulting in the performance gain of the proposed scheme. 

\begin{figure}[t]
\begin{centering}
\includegraphics[scale=0.5]{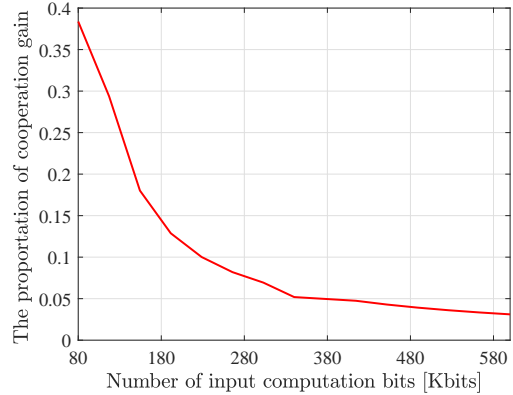}
\vspace{-0.1cm}
 \caption{The proportion of energy saving achieved by cooperation.}\label{fig:Emin_Lu_Larsacl}
\end{centering}
\vspace{-0.1cm}
\end{figure}

Given $L_{\mathrm{h}}=80$ Kbits, the numerical results of the proportion of the cooperation gain over a large range, where the input-data size varies from 80 Kbits to 600 Kbits are presented in Fig. \ref{fig:Emin_Lu_Larsacl}. The proportion of the cooperation gain is defined to be the energy saving achieved by cooperation divided by the total energy consumption of the proposed scheme. We can see that, as the input-data size is increasing, the proportion of the cooperation gain over the total energy consumption is decreasing. The reason is explained as follows: Compared with the schemes without cooperation, the cooperation gain is achieved over the user-to-helper offloading channel. Since the offloading channel capacity is bounded, the cooperation gain is also bounded. Thus when the input-data size is increasing, the proportion of energy saving achieved by cooperation is decreasing.

\begin{figure}[t]
\begin{centering}
\includegraphics[scale=0.5]{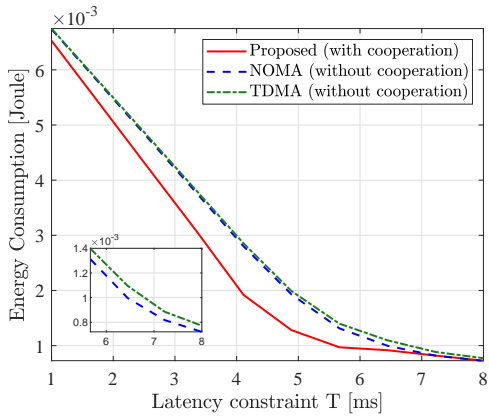}
\vspace{-0.1cm}
 \caption{ The average energy consumption versus $T$. }\label{fig:E_min_T}
\end{centering}
\vspace{-0.1cm}
\end{figure}

The curves of the energy consumption versus the latency constraint are plotted in Fig. \ref{fig:E_min_T}. Here $L_{\mathrm{u}}=L_{\mathrm{h}}=80$ Kbits. We can see that all schemes experience a linear decline in their energy consumption first and gradually become flat when the total transmit time $T$ continues to increase. The optimal scheme saves about $13\%$ energy consumption on average, compared with the schemes without cooperation. In particular, when there is a small $T$, i.e., strict latency requirement, the user and the helper need to offload the input-bits in a high rate to meet the latency constraint, resulting in high energy consumption. Thus a small increase in $T$ can substantially improve the system performance. For looser latency requirement $T$, the performance of the optimal scheme saturates. This is because that in our theoretical analysis, we find that the objective function of Problem (P1) can be rewritten as function $f_{2}(\bm{\ell})$. Note that $f_{2}(\bm{\ell})$ can be obtained by setting $t_{\mathrm{u}}=\alpha T$ and $t_{\mathrm{h}}=(1-\alpha)T$ in $f_{2}(\bm{\ell},\bm{t})$, which is defined in \eqref{eqn:r23}. We observe that $f_{2}(\bm{\ell})$ is an exponential function with respect to $1/T$. Therefore, with longer $T$, the slope of the line with respect to the performance of the optimal scheme becomes smoother.

\subsection{Offloading Data Maximization Case}
\begin{figure}[t]
\begin{centering}
\includegraphics[scale=0.5]{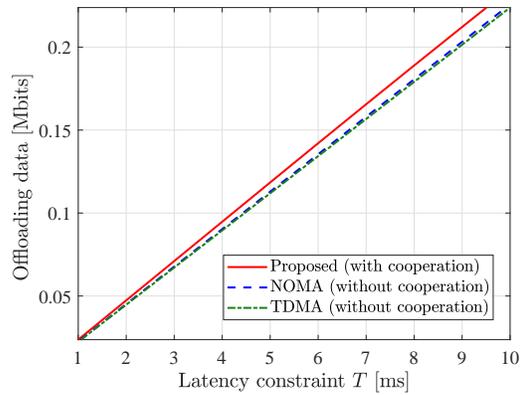}
\vspace{-0.1cm}
 \caption{ The average offloading data versus $T$. }\label{fig:rate_min_T}
\end{centering}
\vspace{-0.1cm}
\end{figure}

Fig. \ref{fig:rate_min_T} gives the results of the offloading data versus the latency constraint $T$. The distance between the user and the helper is fixed as 70 meters. Compared with other benchmark schemes, the proposed optimal scheme achieves the best performance. As $T$ increases, i.e., looser latency constraint, all the schemes' performance improve. It is worth to note that all schemes have approximately linear increase in their offloading data. This trend can be explained from the structure of the offloading data expression. That is, the offloading data is the product of transmission time and rates, i.e., a linear function with respect to the variables $t_{\mathrm{u}}$ and $t_{\mathrm{h}}$. The reasons accounting for this trend in other benchmark schemes are similar to the proposed scheme.

\begin{figure}[t]
\begin{centering}
\includegraphics[scale=0.5]{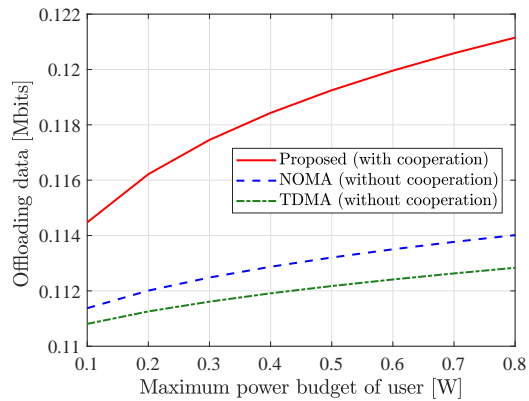}
\vspace{-0.1cm}
 \caption{ The average offloading data versus $\bar{P}_{\mathrm{u}}$. }\label{fig:rate_min_Pu}
\end{centering}
\vspace{-0.1cm}
\end{figure}

With fixed $T=5$ ms, the impact of the maximum transmit power of the user $\bar{P}_{\mathrm{u}}$ on the offloading data is illustrated in Fig. \ref{fig:rate_min_Pu}. In this figure we set the distance between the user and the helper as 70 meters. We can validate the effectiveness of the proposed scheme which provides a gap between itself and other benchmark schemes. The offloading data in all schemes has a logarithmic augment when $\bar{P}_{\mathrm{u}}$ increases. It is intuitive that the transmission rates and the offloading data increase when $\bar{P}_{\mathrm{u}}$  increases. Moreover, the NOMA scheme is slightly better than the TDMA scheme.
\begin{figure}[t]
\begin{centering}
\includegraphics[scale=0.5]{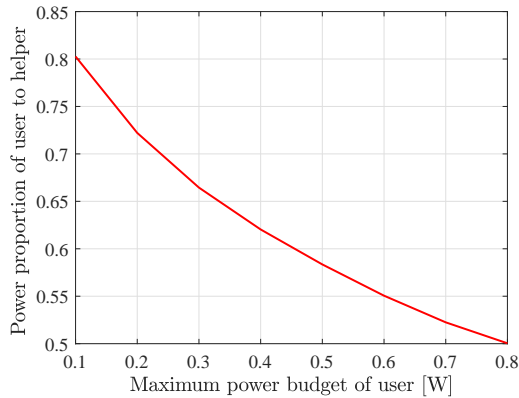}
\vspace{-0.1cm}
 \caption{ Power proportion of user to helper versus $\bar{P}_{\mathrm{u}}$. }\label{fig:rate_min_alp_pa}
\end{centering}
\vspace{-0.1cm}
\end{figure}

Fig. \ref{fig:rate_min_alp_pa} shows the power proportion of the user to the helper, i.e., $\beta$, versus the maximum transmit power constraint in the user $\bar{P}_{\mathrm{u}}$. We fix $T$ as 5 ms and the distance between the user and the helper as 70 meters. With increasing $\bar{P}_{\mathrm{u}}$, the user allocates more power to offload data to the AP. This trend can be explained from \eqref{eqn:r73} where larger $\bar{P}_{\mathrm{u}}$ leads to a smaller $\beta$. This is because, the helper is resource-constrained and needs more time to execute the received data offloaded by the user when the user transmits with higher rate. In this case, due to the latency constraint \eqref{eqn:r69} at the helper, the user prefers to offload more data to the AP.

\begin{figure}[t]
\begin{centering}
\includegraphics[scale=0.5]{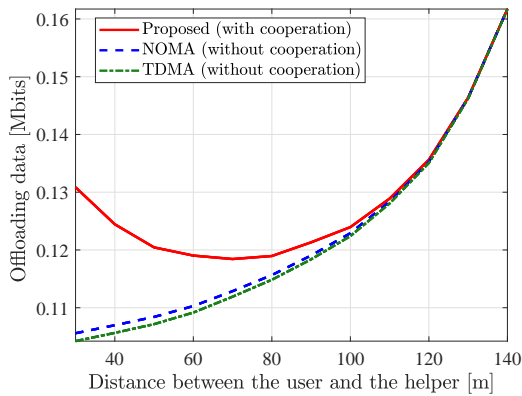}
\vspace{-0.1cm}
 \caption{ The average offloading data versus the distance between the user and helper. }\label{fig:rate_max_distance}
\end{centering}
\vspace{-0.1cm}
\end{figure}
 We investigate the influence of the distance between the user and helper on the offloading data in all schemes, as shown in Fig. \ref{fig:rate_max_distance}. Here the latency constraint $T=5$ ms. We can observe that when the distance between the user and helper is short, or the helper is close to the user, the cooperation yields a large performance gain. While the distance between the user and helper widens, which means that the helper is moving forward to the AP, the performance gaps between the proposed scheme and other benchmark schemes are narrowed. It is because when the helper is moving closer to the AP, the benefit of cooperative computing at the helper cannot compensate the cost of offloading data to the helper. Instead, the user offloads more data  directly to the AP. On the other hand, as the helper is moving closer to the AP, the helper has a better channel for offloading its own data. So in order to maximize the sum offloading data, the system tends to assign longer offloading time duration $t_{\mathrm{h}}$ to the helper. Based on these reasons, as the distance between the user and the helper increases (or the helper is moving closer to the AP), the performance of the proposed cooperation scheme first degrades due to the poorer channel conditions for performing cooperative computing, and then improves since the helper has a better offloading channel. Note that these schemes' performance finally converges to the same value when the helper and AP are at the same location, which has been proved mathematically in Section \ref{sameLoc}. The intuitive explanation is as follows: The proposed cooperation scheme has degree of freedom to utilize the user-to-helper channel for offloading compared with the schemes without cooperation. Thus, if the helper moves to the AP (or is far away from the user), the user-to-helper channel becomes worse and thus the cooperation gain is vanishing. When the helper and AP are at the same location, the user transmits all data to the AP as the AP has more computational capacity. In this case, the proposed cooperation scheme degenerates into the schemes without cooperation. It is worth noting that no matter where the helper is located between the user and the AP, the performance of the proposed cooperation scheme is better than or the same as that of the benchmarks without cooperation.
\begin{figure}[t]
\begin{centering}
\includegraphics[scale=0.5]{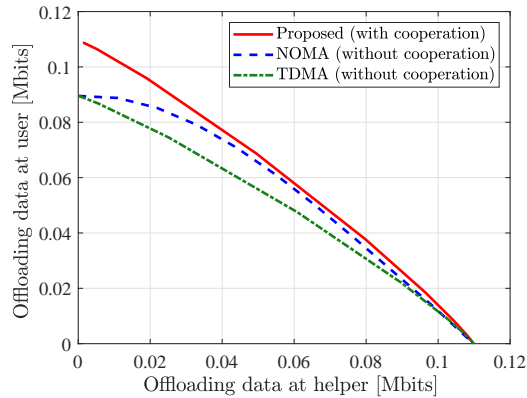}
\vspace{-0.1cm}
 \caption{ Offloading data regions. }\label{fig:rate_rate_region}
\end{centering}
\vspace{-0.1cm}
\end{figure}

Fig. \ref{fig:rate_rate_region} illustrates the offloading data regions of the proposed optimal scheme and the benchmark schemes without cooperation. The distance between the user and helper is fixed as 70 meters and $T=5$ ms. The proposed scheme provides an upper bound for other benchmark schemes. It is observed that when the offloading data in the helper is increasing, the total offloading data in these schemes finally converges to the same values. This is because when the user is not allowed to offload bits, these three schemes have the same offloading data expression and thus have the same optimal values. Note that when the offloading data in the helper is 0 bit, the proposed scheme still outperforms other benchmark schemes. The reason is that in the proposed optimal scheme the user can offload input-data to the helper and AP simultaneously while in the benchmark schemes cooperation computing is not allowed. 
\subsection{Both Cases}

\begin{figure}[t]
\begin{centering}
\includegraphics[scale=0.5]{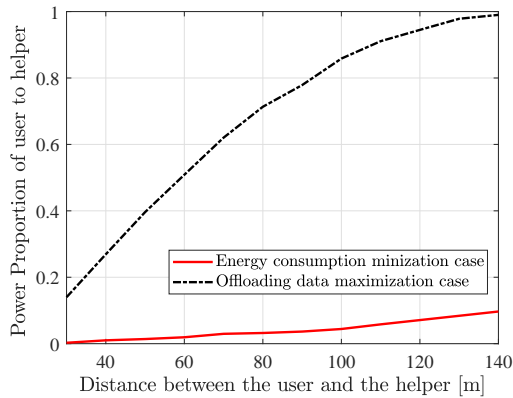}
\vspace{-0.1cm}
 \caption{ Power proportion of user to helper versus the distance between the user and helper. }\label{fig:joint_min_alp}
\end{centering}
\vspace{-0.1cm}
\end{figure}

Fig. \ref{fig:joint_min_alp} illustrates the impact of the distance between the user and helper on the power allocation proportion from the user to the helper, where the helper gradually moves away from the user. The latency $T$ is set to be $5$ ms. We can firstly see when the helper is close to the user, the most transmit power of the user is used for offloading data to the AP. This is because AP is considered to be resource-rich in this paper, thus the user prefers to allocating more transmit power to offload data to the AP. Moreover, we can see that when distance between the user and helper increases, the user allocates more transmit power to offload data to the helper. The reason may be that, in the energy consumption minimization case the user has to complete the computation task under the latency constraint. When the helper is getting away from the user, the user adds the transmit power to the helper so that it can meet the latency constraint. Additionally, note that the helper is located in the middle of the user and the AP and performs SIC. Thus the helper has a higher rate than that of the AP. Thus the user prefers to allocate more transmit power to the helper for rate maximization.

\begin{figure}[t]
\centering
\subfigure[Energy consumption versus $w_{\mathrm{h}}$.]{
\begin{minipage}[t]{1\linewidth}\label{fig:Emin_max_weight1}
\centering
\includegraphics[scale=0.5]{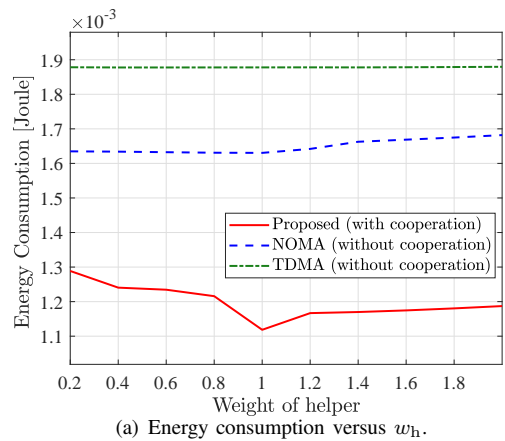}

\end{minipage}%
}%

\subfigure[Offloading data versus $w_{\mathrm{h}}$.]{
\begin{minipage}[t]{1\linewidth}\label{fig:rate_max_weight1}
\centering
\includegraphics[scale=0.5]{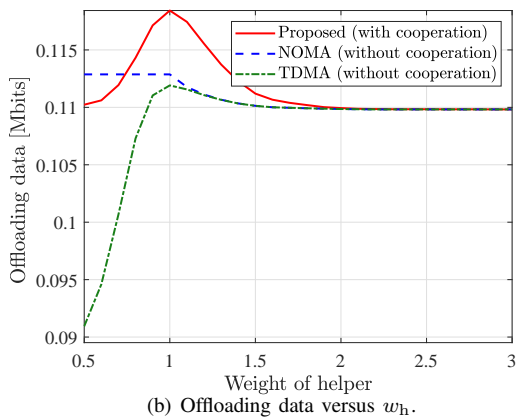}
\end{minipage}%
}%
\caption{The impact of weight factors.}
\end{figure}

 The impact of varying the helper's weight $w_{\mathrm{h}}$ is reflected in Fig. \ref{fig:Emin_max_weight1} and Fig. \ref{fig:rate_max_weight1} for both energy consumption minimization and offloading data maximization, where we fix the user's weight $w_{\mathrm{u}}=1$, $T=5$ ms and the distance between the helper and user as 70 meters. From these figures we can see that when $w_{\mathrm{h}}<w_{\mathrm{u}}$, the performance for both energy consumption minimization and offloading data maximization experiences a growth as the weight increases until $w_{\mathrm{u}}=w_{\mathrm{h}}=1$. When $w_{\mathrm{h}}>w_{\mathrm{u}}$, the performance of all schemes shows a  downward trend. This shows that $w_{\mathrm{u}}=w_{\mathrm{h}}$ is the most advantageous for optimizing the system performance. If one weight is larger than the other one, it takes priority over the other one when allocating system resources, which sacrifices the system performance.

\section{Conclusions}\label{se5}
In this paper, we have proposed a joint communication and computation resource sharing scheme in NOMA-aided cooperative computing system, where the user is enabled to offload task to the helper and the AP simultaneously using NOMA. We have studied the energy consumption minimization problem and offloading data maximization problem, respectively. We have solved the two non-convex problems and obtained some useful insights for practical design as follows.

For the energy consumption minimization problem, first, the user needs to offload more bits to the helper if its per-bit energy consumption of local computing is higher than that of the helper. Otherwise, the user prefers to locally compute more bits. Second, when the helper does not have its own task to execute, the user tries to make its offloading time as long as possible until the helper cannot complete the task offloaded by the user in the remaining time. 

For the offloading data maximization problem,  first, in high SNR, the user spends all transmit power in offloading data to the AP. The optimal strategy is enabling one of the user and helper to occupy the whole transmit time to offload data to the AP. Second, as the distance between the user and the helper increases (or the helper is moving closer to the AP), the performance of the proposed cooperation scheme first degrades due to the poorer channel conditions for performing cooperative computing, and then improves since the helper has a better channel for offloading its own task to the AP.
\begin{appendices}
\section{Proof of Lemma \ref{Lm5}} \label{AP5}
It is obvious that $f_{1}(x)$ is a convex function when the domain of $f_{1}$, i.e., $\textbf{dom} \ f_{1}$, is a convex set. Then to prove the convexity of $f_{2}(\bm{\ell},\bm{t})$, we need to introduce the perspective function of $f_{1}(x)$. Function $g_3$ is defined to be the mapping
  \begin{align}
  	g_{3}: \mathbb{R}^{2}\rightarrow\ & \mathbb{R},\nonumber\\	(x,t)\mapsto\ & tf_{1}(x/t).
  \end{align}
In \cite{Boyd2004}, the authors proved that the perspective operation preserves convexity. Thus $g_3(x,t)$ is a convex function. 

 The feasible region of Problem (P1) with respect to $(\bm{\ell},\bm{t})$ is equivalent to the set 
  \begin{align}
   	S_{1}:=\{(\bm{\ell},\mathbf{t})\in \mathbb{R}^{3}\times \mathbb{R}: \eqref{eqn:r14}, \eqref{eqn:r15}, \eqref{eqn:r16}, \eqref{eqn:r60}\}
  \end{align}
  Note that the constraints \eqref{eqn:r14}, \eqref{eqn:r15}, \eqref{eqn:r16}, and \eqref{eqn:r60} are linear inequalities, implying that $S_{1}$ is a convex set. Note that function $f_{2}(\bm{\ell},\bm{t})$ equals
 \begin{align}
      f_{2}(\bm{\ell},\bm{t})
	=&\frac{w_{\mathrm{u}}}{h_{\mathrm{u,h}}}g_3(\ell_{\mathrm{u,h}}
	+\ell_{\mathrm{u,a}},t_{\mathrm{u}})+\frac{w_{\mathrm{h}}}{h_{\mathrm{h,a}}}g_3(\ell_{\mathrm{h,a}},t_{\mathrm{h}} )
 \nonumber\\&+ \left(\frac{w_{\mathrm{u}}}{h_{\mathrm{u,a}}}
-\frac{w_{\mathrm{u}}}{h_{\mathrm{u,h}}}\right)g_3(\ell_{\mathrm{u,a}},t_{\mathrm{u}} )
\nonumber\\&
+w_{\mathrm{u}}\kappa f_{\mathrm{u}}^{2}(L_{\mathrm{u}} - \ell_{\mathrm{u,h}} - \ell_{\mathrm{u,a}})
\nonumber\\&
+w_{\mathrm{h}}\kappa f_{\mathrm{h}}^{2}(L_{\mathrm{h}}-\ell_{\mathrm{h,a}}+\ell_{\mathrm{u,h}}),
 \end{align}
which is a linear combination of convex functions $g_3(\ell_{\mathrm{u,h}}+\ell_{\mathrm{u,a}},t_{\mathrm{u}})$, $g_3(\ell_{\mathrm{u,a}},t_{\mathrm{u}})$, $g_3(\ell_{\mathrm{h,a}},t_{\mathrm{h}})$, and the linear functions with respect to $\bm{\ell}$. Note that in our system setup we assume that $h_{\mathrm{u,a}}\leq h_{\mathrm{u,h}}$ due to the fact that the helper is closer to the AP than the user. Thus all the coefficients associated with the convex functions $g_3(\ell_{\mathrm{u,h}}+\ell_{\mathrm{u,a}},t_{\mathrm{u}})$, $g_3(\ell_{\mathrm{u,a}},t_{\mathrm{u}})$, and $g_3(\ell_{\mathrm{h,a}},t_{\mathrm{h}})$ are positive. As a result, $f_{2}(\bm{\ell},\bm{t})$ is a convex function with respect to $\{\bm{\ell},\bm{t}\}$ over $S_{1}$.

\section{Proof of Lemma \ref{Lm1}} \label{AP1}
By defining $f_{3}$ to be the mapping 
\begin{align}
f_{3}: \mathbb{R}\rightarrow\ &\mathbb{R},\nonumber\\ x\mapsto\ &(1-x\ln2)2^{x}-1,
\end{align}
we find that $f_{3}(0)=0$. Then the derivative of $f_{3}(x)$ is 
\begin{align}
	\dv{f_{3}}{x}=-(\ln(2))^{2}x2^{x}\leq0,\ \forall x\geq 0.\label{eqn:r24}
\end{align}
From \eqref{eqn:r24}, we know that $f_{3}$ is non-increasing with respect to $x$. Thus $f_{3}(x)\leq 0$ holds for $x\geq 0$.
\section{Proof of Lemma \ref{Lm2}}\label{AP3}
We can prove this lemma by proving $\frac{\partial f_{4}}{\partial t_{\mathrm{u}}}\leq 0$, which is similar to the proof in Appendix \ref{AP1}. Since $f_{4}$ is non-increasing with $t_{\mathrm{u}}\geq 0$, to minimize $f_{4}$ in problem \eqref{eqn:r42}, the constraint $t_{\mathrm{u}}+\frac{\ell_{\mathrm{u,h}}}{f_{\mathrm{h}}} \leq T$ should be activated.
\section{Proof of Lemma \ref{Lm4}}\label{AP4}
The objective function of Problem (P2) equals
\begin{align}
	&w_{\mathrm{u}}t_{\mathrm{u}}\log_{2}\left(\frac{1+p_{\mathrm{u,h}}h_{\mathrm{u,h}}}
	{1+p_{\mathrm{u,h}}h_{\mathrm{u,a}}}\right)
	+  w_{\mathrm{h}}t_{\mathrm{h}}\log_{2}(1 + h_{\mathrm{h,a}}\bar{P}_{\mathrm{h}})
	\nonumber\\&+w_{\mathrm{u}}t_{\mathrm{u}}
	\log_{2}\bigl(1+h_{\mathrm{u,a}}(p_{\mathrm{u,h}} + p_{\mathrm{u,a}})\bigr)  .\label{eqn:r74}
\end{align}
Note that \eqref{eqn:r74} is an increasing function with respect to $p_{\mathrm{u,a}}$. As a result, to maximize the objective function of Problem (P2), the constraint associated with the maximum value of $p_{\mathrm{u,a}}$ \eqref{eqn:r17} needs to be activated, i.e., $p_{\mathrm{u,h}}+p_{\mathrm{u,a}}=\bar{P}_{\mathrm{u}}$. Otherwise, the objective function can be further improved by increasing the value of $p_{\mathrm{u,a}}$ until $p_{\mathrm{u,h}}+p_{\mathrm{u,a}}=\bar{P}_{\mathrm{u}}$. 
\section{Proof of Lemma \ref{lem5}} \label{AP2}
Since the first term in $f_{5}(t_{\mathrm{u}})$ is a linear function with respect to $t_{\mathrm{u}}$ and the second term is a constant, we only need to prove the third term's concavity to further prove that $f_{5}(t_{\mathrm{u}})$ is a concave function. For convenience, defining the mapping $f_{6}$ to be
\begin{align}
f_{6}:\mathbb{R}_{\geq 0}\rightarrow\ &\mathbb{R},\nonumber\\ t_{\mathrm{u}}\mapsto\ &-w_{\mathrm{u}}t_{\mathrm{u}}\log_{2}(1-\frac{h_{\mathrm{u,a}}}{h_{\mathrm{u,h}}}+\frac{h_{\mathrm{u,a}}}{h_{\mathrm{u,h}}}2^{\frac{f_{\mathrm{h}}(T-t_{\mathrm{u}})}{t_{\mathrm{u}}C_{\mathrm{u}}}}),
\end{align}
which is the third term in function \eqref{eqn:r64}, the derivative of $f_{6}(t_{\mathrm{u}})$ with respect to $t_{\mathrm{u}}$ is
\begin{align}
	\dv{f_{6}}{t_{\mathrm{u}}}=
	&-w_{\mathrm{u}}\log_{2}
	\left(1-\frac{h_{\mathrm{u,a}}}
	{h_{\mathrm{u,h}}}
	+\frac{h_{\mathrm{u,a}}}
	{h_{\mathrm{u,h}}}
	2^{\frac{f_{\mathrm{h}}
	(T-t_{\mathrm{u}})}
	{t_{\mathrm{u}}C_{\mathrm{u}}}}\right)
	\nonumber\\&
	+\frac{w_{\mathrm{u}}
	Tf_{\mathrm{h}}h_{\mathrm{u,a}}
	(\ln 2)2^{\frac{f_{\mathrm{h}}(T-t_{\mathrm{u}})}
	{t_{\mathrm{u}}C_{\mathrm{u}}}}}{C_{\mathrm{u}} 
	t_{\mathrm{u}}\left(h_{\mathrm{u,h}}
	-h_{\mathrm{u,a}}
	+h_{\mathrm{u,a}}2^{\frac{f_{\mathrm{h}}
	(T-t_{\mathrm{u}})}{t_{\mathrm{u}}
	C_{\mathrm{u}}}}\right)}.\label{eqn:r51}
\end{align}
To prove the concavity of $f_{6}(t_{\mathrm{u}})$, we further investigate the second derivative of $f_{6}(t_{\mathrm{u}})$ in the following.
\begin{align}
	&\left(\dv{}{t_{\mathrm{u}}}\right)^{2}f_{6}=
	-\frac{w_{\mathrm{u}}(Tf_{\mathrm{h}}\ln 2)^{2}(h_{\mathrm{u,h}}-h_{\mathrm{u,a}})h_{\mathrm{u,a}}2^{\frac{f_{\mathrm{h}}(T-t_{\mathrm{u}})}{t_{\mathrm{u}}C_{\mathrm{u}} }}}{C_{\mathrm{u}}^{2}t_{\mathrm{u}}^{3}(h_{\mathrm{u,h}}-h_{\mathrm{u,a}}+h_{\mathrm{u,a}}2^{\frac{f_{\mathrm{h}}(T-t_{\mathrm{u}})}{t_{\mathrm{u}}C_{\mathrm{u}}}})^{2}}.\label{eqn:r52}
\end{align}
As we mentioned before, because of $h_{\mathrm{u,h}}>h_{\mathrm{u,a}}$, $(\dv{}{t_{\mathrm{u}}})^{2}f_{6}\leq 0$ is satisfied. 

Since the inequality constraints in problem \eqref{eqn:r50}, i.e., $  t_{\mathrm{u}}\geq \frac{Tf_{\mathrm{h}}}{f_{\mathrm{h}}+C_{\mathrm{u}}\log_{2}(1+\bar{P}_{\mathrm{u}}h_{\mathrm{u,h}})}$ and $0\leq t_{\mathrm{u}}\leq T$ are linear inequality constraints with respect to $t_{\mathrm{u}}$, the feasible set of problem \eqref{eqn:r50} is a convex set. Then we can conclude that $f_{5}(t_{\mathrm{u}})$ is concave with respect to $t_{\mathrm{u}}$ over the feasible set of problem \eqref{eqn:r50}.
\end{appendices}

 \begin{footnotesize}
\bibliographystyle{IEEEtran}
 \bibliography{TCOM-TPS-19-0742_R1}
\end{footnotesize}

\end{document}